\documentclass[letterpaper,11pt]{article}
\usepackage{graphicx}
\usepackage{rose-preamble}
\usepackage{macrosProj}

\usepackage[subtle, mathspacing=normal]{savetrees}

\title{History-Independent Load Balancing}
\author{Michael A. Bender\thanks{Department of Computer Science, Stony Brook University.} \and William Kuszmaul\thanks{Computer Science Department, Carnegie Mellon University.} \and Elaine Shi\thanks{Computer Science Department \& Electrical and Computer Engineering, Carnegie Mellon University.} \and Rose Silver\thanks{Computer Science Department, Carnegie Mellon University.}}
\date{}

\begin{document}
\maketitle

\begin{abstract}
We give a (strongly) history-independent two-choice balls-and-bins algorithm on $n$ bins that supports both insertions and deletions on a set of up to $m$ balls, while guaranteeing a maximum load of $m / n + O(1)$ with high probability, and achieving an expected recourse of $O(\log \log (m/n))$ per operation. To the best of our knowledge, this is the first history-independent solution to achieve nontrivial guarantees of any sort for $m/n \ge \omega(1)$ and is the first fully dynamic solution (history independent or not) to achieve $O(1)$ overload with $o(m/n)$ expected recourse. 
\end{abstract}

\section{Introduction}\label{sec:introduction}
Since its introduction more than two decades ago, history independence \cite{NaorTe01, Micciancio97} has become one of the most widely studied security properties in data structures~\cite{Micciancio97, NaorTe01, HartlineHoMo05,HartlineHoMo02,AcarBlHa04, BuchbinderPe03, BlellochGo07, NaorSeWi08, Golovin09, Golovin10,BenderBeJo16}. Formally, a data structure is said to be \defi{history independent} \cite{NaorTe01} (or \defi{strongly history independent} \cite{NaorTe01}) if, for any sequence $t_1, t_2, \ldots$ of times, leaking the data structure's states $S_{t_1}, S_{t_2}, \ldots$ at those times does not leak any information beyond which elements were contained in the data structure at those points in time. With a vanilla data structure, if the adversary hacks into the server and obtains a snapshot of the data structure's state, it can learn not only which elements are currently in the data structure, but also historical information such as the order of past insertions and deletions, or even sensitive information about which elements were present in the past but have since been deleted. History independent data structures defend against such attacks, revealing the information-theoretically minimum amount of information even if the adversary hacks the data structure at multiple points in time.

A natural way to construct a history independent data structure is to define it in such a way that its state at any given moment is \emph{completely determined} by its current set of elements (and its random tape). Hartline et al.~\cite{HartlineHoMo05,HartlineHoMo02} proved that, in most data-structural settings\footnote{Hartline et al.'s result \cite{HartlineHoMo02, HartlineHoMo05} holds whenever the logical states of the data structure form a strongly connected graph (i.e., all logical states are mutually reachable from each other).} (including those in this paper), this is actually the \emph{only} way to achieve history independence. Thus one can think of history independence as really being a type of unique representability. 

In this paper, we revisit a basic problem where efficient history independence has so far proven difficult to achieve: the problem of maintaining a dynamic assignment of balls to bins, where each ball has two random choices for where it could go, and where the goal is to prevent any one bin from being too overloaded. Although this problem has been heavily studied in the data structure \cite{bansal2022balanced, dietzfelbinger2007balanced, frieze2018balanced,naor2008history,PaghRo01} and scheduling \cite{ABKU94,Vocking99,CFMMSU98,mitzenmacher2001power, mitzenmacher1999studying,BCSV00proc,bansal2022balanced} literatures, the best known solutions are all (significantly) history \emph{dependent}. In this paper, we show that history independent solutions can also do surprisingly well. In fact, our final solution even improves \emph{on the best prior history-dependent results}, achieving a doubly-exponential reduction to recourse (the number of balls moved per operation) while maintaining constant overload (the maximum amount by which any bin is overloaded). 

\paragraph{Two-choice load balancing.} In the \defi{two-choice load-balancing problem}, a set $\mathcal{S}$ of up to $m$ balls must be assigned to $n$ bins. Each ball $x \in \mathcal{S}$ has two uniformly random bins $h_1(x), h_2(x) \in [n]$ where it is capable of going. In its most general form, the goal of two-choice load balancing is to maintain a dynamic allocation of balls to bins, so that, over time, as balls are inserted and deleted, we achieve two simultaneous guarantees:
\begin{enumerate}
\item \textbf{Low Overload:} The amount by which the fullest bin’s load exceeds $m/n$ (i.e., the \defi{overload}) is small;
\item \textbf{Low Recourse: } On any given insertion and deletion, the number of balls that get moved around (this is known as the \defi{recourse}) is small.  Recourse will often be measured as a function of $\mu \coloneqq m/n$.
\end{enumerate}

The two-choice load-balancing problem \cite{ABKU94,Vocking99,CFMMSU98,mitzenmacher2001power, mitzenmacher1999studying,BCSV00proc,bansal2022balanced} can be viewed as capturing a natural scheduling problem: balls represent jobs, each of which are capable of running on two random servers.\footnote{In practice, there are many reasons, depending on the setting, that a system may choose to restrict each job to only two possible servers. In general, it allows the system to minimize the degree to which it duplicates certain types of resources. As a simple example, suppose that a job requires access to a stream of incoming data that is specific to that job. If the data is sent by external sources who do not know which server the job will be on, they can send the data to just two servers, instead of all $n$ servers.} The goal is to schedule jobs to servers so that no server is too overloaded and so that, as jobs arrive and depart, the jobs that are already present experience as little migration as possible.

As we will discuss, in some restricted versions of the problem (e.g., where balls are inserted but not deleted), it is even possible to achieve nontrivial guarantees with no migration at all. But, in the full version of the problem, where balls are both inserted and deleted, there is strong evidence \cite{bansal2022balanced} that any algorithm with low overload must incur recourse.

The simplest solution to the two-choice load-balancing problem is to assign each ball $x \in \mathcal{S}$ to its \emph{first} choice $h_1(x)$. This solution, known as the \defi{single-choice algorithm}, has two appealing properties: it is naturally history independent, and it has no recourse. However, the algorithm has a relatively large overload, e.g., when $m \ge n \log n$, the overload is $\Theta(\sqrt{m \log n})$ with high probability. 

\paragraph{Past work: The (history-dependent) power of two choices.} By using both choices, instead of just one, one can get much better bounds on overload.

If we consider a setting in which balls are inserted (but not deleted), then better overload can be achieved with the \defi{greedy algorithm}, which implements insertions by placing the new ball $x$ in the less loaded of the two bins $h_1(x)$ and $h_2(x)$. This simple algorithm, which was first analyzed for $m = n$ by Azar, Broder, Karlin and Upfal~\cite{ABKU94} in 1994, ends up being surprisingly tricky to analyze for larger $m$, and it was only after a long line of work \cite{ABKU94,Vocking99,CFMMSU98,mitzenmacher2001power, mitzenmacher1999studying,BCSV00proc} that Berenbrink, Czumaj, Steger, and V{\"o}cking~\cite{BCSV00proc} were able to achieve a tight bound, showing that the algorithm achieves overload $\log \log n + O(1)$ with high probability in $n$ (and independently of $m$). 

Whereas the single-choice algorithm was history independent, the greedy algorithm is emphatically not---changes to the order in which elements are inserted can result in significantly different outcomes for the final state of the system. The algorithm also suffers from the fact that its overload guarantee, which holds for insertion-only workloads, actually fails for workloads with both insertions and deletions \cite{bansal2022balanced}. In fact, Bansal and Kuszmaul \cite{bansal2022balanced} give evidence (in the form of a lower bound for a large class of algorithms) that the only way to achieve good load balancing on fully dynamic workloads is with algorithms that incur non-zero recourse. Note that, if one is willing to incur recourse, then one can use \emph{tombstones}\footnote{This means that the algorithm simply \emph{marks} elements as deleted, and then rebuilds the entire system (incurring significant recourse) every, say, $n$ operations.} in order to transform the greedy algorithm into a dynamic algorithm---this maintains an overload bound of $O(\log \log n)$ while achieving amortized expected recourse $O(\mu) = O(m/n)$ per operation. 

In the data-structures literature, there has also been a great deal of work on non-greedy solutions that achieve even better bounds on overload \cite{PaghRo01, dietzfelbinger2007balanced, frieze2018balanced}. Notably, Dietzfelbinger and Weidling \cite{dietzfelbinger2007balanced} give an (again, highly history dependent) solution that achieves overload $\le 1$ with high probability, while achieving expected recourse $O(\mu)$.\footnote{Although the expected recourse in \cite{dietzfelbinger2007balanced} is $O(\mu)$, the expected \emph{time} ends up being much larger (super-polynomial in $\mu = m/n$). We remark that this will not be the case for any of the algorithms studied in this paper---all of the algorithms that we propose can be straightforwardly implemented to incur time proportional to their recourse, and with $O(m)$ space used to store metadata.}  Whether  or not a solution with overload $O(1)$ and expected recourse $o(\mu)$ exists has, as far as we know, remained open.

\paragraph{Past work: History independent solutions. }For history-independent solutions, the only parameter regime that is well understood is the one where $m < n/2 - \Omega(n)$ (i.e., $\mu = 1/2 - \Omega(1)$) \cite{naor2008history}. Here, Naor, Segev, and Wieder \cite{naor2008history} show that it is possible to achieve both overload and expected recourse $O(1)$. This is achieved by applying a separate orientation algorithm to each connected component of the graph $(V = [n], E = \{(h_1(x), h_2(x)) \mid x \in \setBalls\})$, and by exploiting the fact that most of the components in the graph contain only a constant number of edges. This ``small-component’’ graph structure is special to the $m < n/2 - \Omega(n)$ regime, and the question of whether efficient history-independent solutions also exist for larger $m$ has remained open.

\subsection{This Paper: The History-Independent Power of Two Choices}
In this paper, we construct a history-independent two-choice algorithm that simultaneously achieves overload $O(1)$ with high probability in $n$, and expected recourse $O(\log \log \mu)$ where $\mu = m/n$. Even for non-history-independent solutions, our algorithm is the first to achieve overload $O(1)$ with recourse $o(\mu)$.

We remark that, due to the aforementioned unique representability property shown by Hartline et al.~\cite{HartlineHoMo05,HartlineHoMo02}, we can fully define any history-independent data structure by simply describing the data structure's state given a set $\mathcal{S}$ of balls. This is the approach that we will take in discussing the algorithms below.

\paragraph{Warmup: History-independent greedy.} We begin in Section~\ref{sec:greedy} by exploring what is arguably the simplest history-independent algorithm that one could consider (after single-choice), namely, the \defi{history-independent greedy algorithm}. The algorithm calculates the allocation for a given set $\mathcal{S}$ of balls by assigning a canonical ordering $x_1 < x_2 < \cdots$ to the balls in $\mathcal{S}$, and then computing the allocation that the greedy algorithm would have produced if the balls were inserted in that order $(x_1, x_2, \ldots)$. 

For many data structural problems \cite{naor2008history,berger2022memoryless,kuszmaul2023strongly,AmbleKn74,seidel1996randomized}, this basic approach (i.e., applying a canonical ordering, and then inserting the items greedily in that order) ends up being an efficient or even optimal history-independent solution for that problem. This makes the algorithm a natural first candidate for our exploration.

We show that the history-independent greedy algorithm \emph{does} achieve a nontrivial bound: it incurs expected recourse $O(\mu)$ while (trivially) offering an overload bound of $\log \log n + O(1)$. We also prove for a large range of parameters this recourse analysis is tight.

\paragraph{A better algorithm: Slice and Spread.} The main result of the paper is an alternative history-independent algorithm that achieves a doubly exponentially smaller expected recourse of $O(\log \log \mu)$ while offering a high-probability overload of $O(1)$. 

We begin in Section \ref{sec:slicespread} by presenting a basic version of the algorithm, which we call the \defi{Slice and Spread Algorithm}, that achieves the desired recourse bound of $O(\log \log \mu)$ but that comes with only a relatively weak overload bound. All that the basic version of the algorithm guarantees for overload is that the \emph{average} number of balls above height $m/n$ in each bin is $O(1)$ —  but the algorithm does not (a priori) say anything about the maximum.

We then show in Sections \ref{sec:postprocessoverload} an (almost) black-box approach for history-independently transforming an algorithm with low expected overload \emph{per bin} into an algorithm with low maximum overload \emph{across all the bins}, while only increasing expected recourse by a constant factor. This leads to the final result of the paper:
\begin{restatable}{thm}{thmfulltheorem}\label{thm:full-theorem}
    There exists a history-independent two-choice allocation algorithm $\alg$ that achieves the following guarantees: The expected recourse of $\alg$ is $O(\log\log (m/n))$. Furthermore, for each set $\setBalls$ of balls, the overload induced by $\alg$ is $O(1)$ with high probability in $n$. 
\end{restatable}

\paragraph{The surprising power of history independence. } Our final result can be viewed as part of a larger trend in recent years, in which history-independent data structures have been able to achieve \emph{better overall bounds} than the prior (non-history-independent!) states of the art \cite{kuszmaul2023strongly,bender2024online,behnezhad2019fully}.  These results suggest that history independence should be viewed not just as a natural privacy guarantee, but also as a powerful algorithmic paradigm for building natural and efficient algorithms.
\section{Preliminaries}\label{sec:preliminaries}

\paragraph{Dynamic balls and bins.} We begin by reviewing the dynamic balls-and-bins setting. 
In a balls-and-bins system, there are a fixed number $n$ of bins. Balls are \defi{inserted} into and \defi{deleted} from the system dynamically over time. There are up to $m$ balls in the system at any time. When a ball is inserted into the system, it is allocated to a bin. When a ball is deleted from the system, it is removed from its bin. Additionally, as balls get inserted or deleted, balls that are currently in bins can get shifted around to other bins. The number of balls that get shifted during a given insertion or deletion is referred to as the \defi{recourse} of the allocation algorithm during that operation.

Balls come from the universe 
$\universe$, and we denote $\setBalls_t \subseteq \universe$ as the set of balls currently in the system at time step~$t$. When there is no ambiguity, we abbreviate to $\setBalls$. An adversary can update the set of balls present in the system by inserting or deleting one ball per time step.

The \defi{state of the system} is (1) a set $\setBalls \subseteq \universe$ of balls, (2) a set of $n$ bins labeled $1,\ldots,n$, and (3) an allocation of each ball $x \in \setBalls$ to exactly one of the $n$ bins. Given a state, the \defi{load} $\ell_i$ of a bin~$i$ is the number of balls allocated to bin~$i$; the \defi{maximum load} is the maximum over all bins of the loads of the bins; the \defi{overload} is the amount by which the maximum load exceeds $m/n$. We will also refer to the \defi{cumulative overload}, which is the total amount by which the bins' loads exceed $m/n$ (i.e., $\sum_{i=1}^n \max(\ell_i - m/n,0)$).

\paragraph{The two-choice paradigm.} In the two-choice paradigm, the system gets two fully random hash functions $\hashOne$, $\hashTwo$: $\universe \rightarrow [n]$. 
The system maintains the invariant that, for all balls $x$ present in the system, $x$ is allocated either to bin $\hashOne[x]$ or $\hashTwo[x]$, which are the two \defi{allowable locations for $x$}. In particular, in each time step $t$, the \defi{allocator} assigns (or reassigns) each ball $x \in \setBalls_t$ to one of its two allowable locations. The allocator may be randomized, in which case, we can model the system as having an additional infinite string $R$ of random bits (distinct from the random hash functions $\hashOne$ and $\hashTwo$).

The goal is to design an allocation strategy in the two-choice setting that achieves a good high-probability bound on the overload for each state, while simultaneously minimizing the expected number of ball movements per insert/delete.

\paragraph{History-independent two-choice allocation.} 
A two-choice allocation algorithm is history independent if, for each $\setBalls$, $\hashOne$, $\hashTwo$, and $\random$, the allocator allocates all $x \in \setBalls$ to bins according to some function $\allocation<\setBalls,\hashOne,\hashTwo,\random> : \setBalls \rightarrow [n]$. When the context is clear, we will drop the subscripts $\hashOne$, $\hashTwo$, and $\random$ and simply write $\allocation$.

Since our allocation is history independent, it is not a function of time but rather only a function of $\setBalls$ (and additional randomness). When analyzing the overload of the algorithm, we'll want to show that, for any set $\setBalls$, the overload is small with high probability in $n$, where the randomness is over $\hashOne$, $\hashTwo$, and $R$.

\paragraph{Recourse for history independent solutions.} Let two sets $\setBalls, \setBallsPrime \subseteq \universe$ be \defi{neighboring sets} if their contents differ by exactly one ball, i.e. $|\setBalls \triangle \setBallsPrime| = 1$.\footnote{The symmetric difference of sets $A$ and $B$ is defined as $A \triangle B \coloneqq (A \setminus B) \cup (B \setminus A)$.} 
For any pair $(\setBalls,\setBallsPrime)$ of neighboring sets, and for any $\hashOne$, $\hashTwo$, $\random$, we are interested in the \defi{recourse} of $\Paren{\setBalls,\setBallsPrime,\hashOne,\hashTwo,\random}$, which is the number of balls whose allocations change between $\setBalls$ and $\setBallsPrime$. More formally, it is defined as
\begin{equation}
\recourse\Paren{\setBalls,\setBallsPrime,\hashOne,\hashTwo,\random} \coloneqq
    1 + \sum_{x \in \setBalls \cap \setBallsPrime} \ind\Paren{\allocation<\setBalls,\hashOne,\hashTwo,\random>[x] \ne \allocation<\setBallsPrime,\hashOne,\hashTwo,\random>[x]}.\footnote{\text{$\ind$ is the indicator function that maps true to $1$ and false to $0$.}}
\end{equation}
The expected recourse $\recourse[\alg]$ of an allocation algorithm $\alg$ is defined as:
\begin{equation}\label{eq:max-exp-recourse}
\recourse[\alg] \coloneqq \max_{\substack{\setBalls, \setBallsPrime \\ |\setBalls \triangle \setBallsPrime| \le 1}} \E_{\hashOne,\hashTwo,\random} [\recourse(\setBalls, \setBallsPrime, \hashOne, \hashTwo, \random)].
\end{equation}

\paragraph{A WLOG assumption on the number of balls. }We conclude the section by noting one WLOG assumption that will be helpful in some sections of the paper. 

Note that, by taking each set of size $k < m$ balls and adding $(m - k)$ \defi{dummy balls} $d_1, \ldots, d_k$ (these are fake balls to pad the size of the set), we can treat any set of size less than $m$ as a set of size $m$. Furthermore, by alternating between $m$ and $m - 1$ balls, we can get between any two sets of size $m$. Therefore, when designing history-independent allocation algorithms, we may assume without loss of generality that all of our sets have sizes $m - 1$ or $m$, and, in particular, that neighboring sets $\setBalls$ and $\setBallsPrime$ have sizes $m-1$ and $m$.
\section{The History-Independent Greedy Algorithm}\label{sec:greedy}
We now formally describe the history-independent greedy allocation strategy, which we call
\defi{HI Greedy}. For each $\setBalls$, $\hashOne$, $\hashTwo$, and $\random$, the HI Greedy allocator $\allocation<\setBalls,\hashOne,\hashTwo,\random>$ can be computed via the following procedure: We imagine that all of the bins are initialized as empty, and that the procedure inserts all of the balls $x \in \setBalls$ one-by-one into the system according to some fixed canonical ordering of the balls. When each $x \in \setBalls$ is inserted, $x$ is then allocated to the bin in $\{\hashOne[x], \hashTwo[x]\}$ that has the smaller load at that moment (ties are broken by using $\hashOne$).

The following proposition describes both the recourse and overload of HI Greedy.
\begin{proposition}
    The expected recourse of HI Greedy is $O(m/n)$. Furthermore, for each set $\setBalls$, with high probability in $n$, the overload is $\log\log n + O(1)$.
    \label{prop:greedy}
\end{proposition}
\begin{proof}
The overload being $\log \log n + O(1)$ is directly
implied by Berenbrink, Czumaj, Steger, and V{\"o}cking~\cite{BCSV00proc}. Below
we focus on proving the recourse.

Fix two arbitrary neighboring sets $\setBalls$ and $\setBalls' = \setBalls \cup \{x^*\}$ that differ in only one ball denoted $x^*$.
Consider two parallel worlds--called world $0$ and world $1$--in which we insert the balls in $\setBalls$ and $\setBalls'$, respectively, 
one-by-one according to the canonical ordering. We imagine that we perform one insertion in each of the two worlds in each time step, and we will analyze the difference in the bin loads in the two worlds in each time step.

Let $t^*$ be the time step when $x^*$ is inserted into world $1$. (We may assume that nothing happens in world $0$ at this time step.) For each time step before $t^*$, the two worlds have identical configurations. However, at time step $t^*$, exactly one bin denoted $i^*$ differs in these two worlds. Specifically, bin $i^*$ picks up one more ball in world $1$.  

The key to bounding recourse is to observe the following fact:

\begin{fact}
At the end of every time step $t \geq t^*$, 
exactly one bin differs in load in the two worlds, and 
in this special bin, world $1$ has one more ball than world $0$. 
\label{fact:special}
\end{fact}
\begin{proof}
We can prove this fact inductively for each $t \ge t^*$, where we have just seen that it is true when $t = t^*$.

Suppose the claim is true for some time step $t = k$. We show that the claim is then true for time step $t = k+1$ (assuming the time step exists). When we insert the ball $x$ in both worlds in time step $k+1$, either $(1)$ $x$'s allocation is the same in both worlds, or $(2)$ $x$'s allocation is different in both worlds. In case $1$, the special bin from time step $k$ remains a special bin in time step $k+1$, and furthermore remains the only special bin, thus proving the claim for time step $k+1$ in this case.

In case $2$, suppose that $x$ was allocated to some bin $i_0$ in world $0$ and is allocated to some different bin $i_1$ in world $1$. Note that this is only possible if the special bin is $i_0$. We prove the claim by imagining that in time step $k+1$ we first insert $x$ into world $0$ and then insert $x$ into world $1$. When we first insert $x$ into world $0$, the load of bin $i_0$ in world $0$ catches up to its load in world $1$; in other words, none of the bins at this point are special anymore. When $x$ is then inserted into world $1$, bin $i_1$ in world $1$ gets an extra ball and becomes the special bin between the two worlds, thus proving the claim.
\end{proof}

\paragraph{Completing the proof of Proposition~\ref{prop:greedy}.} When a ball is inserted, it is only possible for it to enter different bins in the two worlds if its two choices include the differing bin in that time step.
The probability that it happens to choose
the differing bin is $O(1/n)$.
Therefore, among all $\le m$ balls inserted, in expectation
$O(m/n)$ of them will enter different bins in the two worlds. Since these are the only balls that contribute to the recourse, the expected recourse
is upper bounded by $O(m/n)$.
\end{proof}

We also prove a lower bound on the expected recourse of HI Greedy, showing that so long as $m = n^{2 - \Omega(1)}$, the expected recourse is \emph{at least} $\Omega(m/n)$. For brevity, the proof is deferred to Appendix \ref{sec:greedylower}.
\begin{restatable}{prop}{propgreedylower}
    If $m \le n^{2 - \Omega(1)}$, then the expected recourse of HI Greedy is $\Omega(m/n)$.
\label{prop:greedylower}
\end{restatable}

\section{Improving the Overload: Slicing and Spreading}\label{sec:slice-and-spread}
In this section, we introduce the Slice and Spread Algorithm (Algorithm~\ref{alg:slice-and-spread}), whose guarantees are given in Theorem~\ref{thm:offline-slice-and-spread}. 
\begin{restatable}{thm}{thmofflinesliceandspread}\label{thm:offline-slice-and-spread}
    The expected recourse of Slice and Spread is $O(\log\log (m/n))$. Furthermore, for each state, the cumulative overload of that state under Slice and Spread is $O(n)$ with high probability in $n$.
\end{restatable}

Throughout the section, we will make the WLOG assumption from Section~\ref{sec:preliminaries} that the number of balls is always either $m$ or $m - 1$. 

\subsection{Slice and Spread Overview}
\paragraph{The slice-and-spread gadget.} At the heart of Slice and Spread is a key gadget that we will refer to as the \defi{slice-and-spread gadget}. The gadget first takes as input an initial allocation of balls to bins, where each ball $x$ is placed into $\hashOne[x]$. The gadget works by first picking some \defi{slicing threshold} $\ell$. Then, for all bins that have more than $\ell$ balls, the gadget tries to temporarily remove balls from each of these bins so that each bin's load is reduced to $\ell$. (For reasons we will come back to later, this will sometimes fail.) One can imagine the gadget as cleanly slicing off the top of each bin so each bin has height $\ell$. Finally, the gadget reallocates the removed balls, putting each removed ball $x$ into $\hashTwo[x]$. This is the ``spreading'' part, as one can imagine the gadget spreading out the removed balls among the bins.

\paragraph{Reducing the overload.} The slice-and-spread gadget can be used to reduce the overload. To see how, imagine the scenario in which $m >> n$ balls are allocated into $n$ bins according to their first-choice hash. By a Chernoff bound, with high probability in $n$, this allocation has overload $x = O\Paren{\sqrt{m/n} \cdot \sqrt{\log n} + \log n}$ and furthermore, every bin has load between $m/n \pm x$ with high probability in $n$. Suppose the gadget picks the slicing threshold to be approximately the height $m/n - x$ (the final algorithm will need to pick the threshold more carefully). Then, the gadget will remove roughly $O(xn) = O(\sqrt{mn})$ balls from the system and put them into their second hash. Since the second hashes are independent and uniformly random, then this is equivalent to throwing $\sqrt{mn}$ balls into $n$ bins. The new overload is reduced to $\tilde{O}\Paren{\sqrt{x}}$.

\paragraph{Repeatedly applying the slice-and-spread gadget.} The algorithm repeatedly applies the slice-and-spread gadget in order to continually reduce the overload. The algorithm proceeds in rounds, where in each round, the gadget carefully selects a slicing threshold so that intuitively most of the jaggedness from the bins gets sliced and then spread. However, in order for the slice and spread to reduce the overload, the gadget must be careful to only slice away balls that have not been previously sliced before. This is so that every time balls are spread, they are making use of the fresh randomness coming from their $\hashTwo$. To do this (and, specifically, to do this with good recourse), the algorithm gives balls in $\setBalls$ a round assignment, where a ball can only be evicted in the round if it is also assigned to that round. The details of the round assignment are given in Section~\ref{sec:slice-and-spread-description}. With the round assignments, and with carefully-chosen slicing thresholds, the algorithm is able to reduce the cumulative overload to $O(n)$.

\paragraph{Gadget failures.} Before we present the algorithm in its entirety, it is worth identifying what will end up being one of the primary challenges in the analysis of its overload. In some cases, a bin may contain many balls, but very few (or no) balls assigned to the current round. In this case the slice-and-spread gadget will not be able to bring the load down to the threshold for the current round. These types of \emph{slicing failures} can be avoided in rounds of the algorithm where the number of sliced balls is very large, but become unavoidable in rounds where the number of such balls is $o(n \log n)$. 

Somewhat subtly, \emph{spreading failures} also turn out to be a problem. This happens when the slicing gadget fails to place enough balls in a bin for it to \emph{reach} the threshold in the next round of the algorithm. Spreading failures in one round increase the number of balls that get sliced in the next, and can prevent the number of balls in play from decreasing at the necessary rate between consecutive rounds. 

A key insight in the algorithm is that, if we pick our thresholds correctly, then even if we allow slicing and spreading failures to occur, we can still bound the final number of balls that reside above height $m/n$. At the same time, we will be able to keep the number of rounds that our algorithm has to be $O(\log \log (m/n))$, which will in turn allow us to bound the expected recourse.

\subsection{Slice and Spread Full Description}\label{sec:slice-and-spread-description}
The Slice and Spread Algorithm begins by initializing parameters. These parameters are assumed in the Slice and Spread procedure described in Algorithm~\ref{alg:slice-and-spread}.

\paragraph{Variable initialization.} Let $\mu \coloneqq m/n$ and let $T \coloneqq \log_{4/3} \log (\mu)$. We refer to each $t \in [T]$ as a \defi{round}. For each round $t \in [T]$, the algorithm initializes the following:
\begin{tcolorbox}[colback=gray!20, colframe=gray!60, sharp corners]
\begin{itemize}
    \item $\mu_t \coloneqq (\mu_{t - 1})^{3/4}$, where $\mu_0 \coloneqq \mu$.
    \item $m_t \coloneqq \mu_t \cdot n$, where $m_0 \coloneqq m$. 
    \item $\threshold{t} \coloneqq \mu - \mu_t$, where $\threshold{0} \coloneqq 0$.
\end{itemize}
\end{tcolorbox}
To a first approximation, we can think of $\mu_t$ as being the expected number of balls each bin receives during the spreading stage of each round $t$ if we start with $|\setBalls| = m$ and if we never incur any gadget failures; we can also think of $m_t$ as the number of balls that are thrown in each round $t$ in this scenario. The values $\threshold{t}$ will be the heights at which the algorithm tries to slice each bin during round $t$, and we will refer to them as \defi{slicing thresholds}. 

We remark that, to simplify our exposition throughout the section, we will treat each $\mu_t$ as an integer---this allows us to avoid having to propagate a large number of floors and ceilings throughout the section. 

\paragraph{Round assignments.} Using $\random$, the algorithm also initializes each ball $x \in \setBalls$ to be a \defi{round-assigned ball} with probability $0.01$. (We remark that the reason for assigning balls to rounds with probability $0.01$ becomes apparent only in Section~\ref{sec:postprocessoverload}.) The algorithm makes the additional initializations related specifically to round-assigned balls:
\begin{tcolorbox}[colback=gray!20, colframe=gray!60, sharp corners]
    \begin{itemize}
        \item $\cmap: \universe \rightarrow [\maxRounds]$ maps each round-assigned ball $x \in \setBalls$ to a round $\cmap[x]$ such that each ball independently satisfies  $$\Pr[\cmap[x] = t] = \Theta\Paren{\frac{m_{t - 1}}{m}}$$ for each $t \in [\maxRounds]$. 
        \item $\colorSet[t][i] \coloneqq \{x \in \setBalls \mid \hashOne[x] = i \text{ and } \cmap[x] = t\}$ for each $(t,i) \in [T] \times [n]$.
    \end{itemize}
\end{tcolorbox}
We now give the full details of Slice and Spread below in Algorithm~\ref{alg:slice-and-spread}.

\begin{algorithm}[H]
\caption{Slice and Spread.}
\label{alg:slice-and-spread}
\begin{algorithmic}[1]
    \Procedure{SliceAndSpread}{$\setBalls, \hashOne, \hashTwo, \random$}
        \State Place each $x \in \setBalls$ into bin $\hashOne[x]$.
        \For{each round $t = 1$ to $T$}
            \State \label{line:slice}\defi{Slice:} For each $i \in [n]$, let $\load{t}{i}$ be the current number of balls in bin $i$. For each $i \in [n]$, take the \\\hspace{60pt}$\min(|\colorSet[t][i]|,\load{t}{i} - \threshold{t})$-first balls from $\colorSet[t][i]$ according to a canonical total ordering, remove \\\hspace{60pt}(evict) them from bin $i$, and add them to the \defi{rethrow set} $\rethrow{t}$.
            \State \label{line:spread}\defi{Spread:} For each ball $x \in \rethrow{t}$, place $x$ into bin $\hashTwo[x]$.
        \EndFor
    \EndProcedure
\end{algorithmic}
\end{algorithm}

\subsection{The Cumulative Overload of Slice and Spread}
In this section, we bound the cumulative overload of Slice and Spread:

\begin{restatable}{prop}{propfoo}\label{prop:foo}
    For each $\setBalls \subseteq \universe$, if $|\setBalls| \le m$, then after applying Slice and Spread to $\setBalls$, the cumulative overload is $O(n)$ with high probability in $n$. 
\end{restatable}

To show that the cumulative overload is small, we want to be arguing that Slice and Spread is actually ``smoothing out'' the bin loads in every round. Ideally, in each round, it is taking balls from bins with large loads and distributing them to bins with small loads. 

There are a couple of situations which make it difficult for Slice and Spread to smooth out the loads. For example, during the slicing stage in a round $t$, there may be bins that cannot be sliced down to $\threshold{t}$. This is because of at least one of two reasons: either the bin doesn’t have enough round-$t$ assigned balls in its possession to support the removal of balls, or the bin’s load may be exceedingly high. Additionally, there may be bins whose loads are already below the slicing threshold $\threshold{t}$. This implies that an earlier round of Slice and Spread failed to smooth out balls into this bin.

It’s important to note that overfilled and underfilled bins are connected to each other:  having a lot of overfilled bins means that there may be more underfilled bins, and vice versa. More generally, there can be cause-and-effect loops between rounds. If a bin is overloaded (resp.~underloaded) in one round, it is more likely to continue being overloaded (resp.~underloaded) in the next. And, more subtly, if there are many underfilled bins in one round, then this causes more balls to be above height $\threshold{t}$ in that round (namely, the balls that should have been below the threshold in the underfilled bins), which causes more balls to get sliced and spread in that round, which causes more overloaded bins in the next round. Bounding the effect of these types of feedback loops, and more generally bounding the effect of overfilled and underfilled bins on the algorithm, is the main technical challenge in this section. 

To prove Proposition~\ref{prop:foo}, we will make use of some additional notation and definitions. To simplify our indexing of variables, throughout the section, we will consider there to be a \defi{round 0} in which there is a spreading stage (given by Line~\ref{line:spread} in Algorithm~\ref{alg:slice-and-spread}) but no slicing stage. 

With this convention in mind, for $t \in \{0, \ldots, T\}$ and $i \in [n]$, define $\loadi{t}{i}$ to be the number of balls in bin $i$ immediately prior to the spreading stage in round $t$. (For $t = 0$, this gives $\loadi{t}{i} = 0$.)  A bin is \defi{underfilled bad in round $t$} if $\loadi{t}{i} < \threshold{t}$ and is \defi{overfilled bad in round $t$} if $\loadi{t}{i} > \threshold{t}$. Define the \defi{round-$t$ error} of bin $i$ to be $|\loadi{t}{i} - \threshold{t}|$. Define $\errunder{t}$ (resp.~$\errover{t}$) to be the sum of the errors of the underfilled bad (resp.~overfilled bad) bins in round $t$, and define $\err{t} \coloneqq \errunder{t} + \errover{t}$.

We begin our analysis with a technical lemma about throwing balls into bins. We remark that the proof of this lemma makes use of McDiarmid's Inequality \cite{mcdiarmid1989method} (which we restate for reference in Appendix~\ref{app:mcdiarmid}).
\begin{lemma}
Let $\mu \ge 1$ and suppose we throw $m = \mu n \pm O(n)$ balls into $n$ bins uniformly at random. Let $X_i$ denote the number of balls in each bin $i$, and call bin $i$ \defi{imbalanced} if $|X_i - \mu| \ge \mu^{3/4}$. Define $A_i$ to be $0$ if bin $i$ is balanced, and to be $|\mu - X_i|$ if bin $i$ is imbalanced. Then, with high probability in $m$, the number of imbalanced bins is at most $n / \mu^{\omega(1)}$ and $\sum_{i = 1}^{n} A_i \le n / \mu^{\omega(1)}$, where the $\omega(1)$ term is parameterized by $\mu$ (not $n$).
\label{lem:imbalance}
\end{lemma}
\begin{proof}
Since the number of imbalanced bins is at most $\sum_{i = 1}^n A_i$, it suffices to show that $\sum_{i = 1}^{n} A_i \le n / \mu^{\omega(1)}$. Furthermore, since the result is trivially true for small $\mu = O(1)$, we may assume that $\mu$ is at least a large positive constant, and therefore also that $m \ge n$.

To begin the analysis, let us bound $\E[A_i]$. By a Chernoff bound, we have for all $j \ge \sqrt{h}$ that
\begin{equation}\Pr[|X_i - \mu| \ge j] \le e^{-\Omega(j/\sqrt{\mu})}.
\label{eq:poorchernoff}
\end{equation}
We can therefore bound
\begin{align*}
    \E[A_i] & \le \mu^{3/4} \cdot \Pr[A_i \neq 0] + \sum_{j > \mu^{3/4}} \Pr[A_i \ge j] \\
            & = \mu^{3/4} \cdot \Pr[|X_i - \mu| > \mu^{3/4}] + \sum_{j > \mu^{3/4}} \Pr[|X_i - \mu| \ge j] \\
            & \le \mu^{3/4} \cdot e^{-\Omega(\mu^{1/4})} + \sum_{j > \mu^{3/4}} e^{-\Omega(j / \sqrt{\mu})} \tag{by \eqref{eq:poorchernoff}} \\
            & \le \mu^{3/4} \cdot e^{-\Omega(\mu^{1/4})} + e^{-\Omega(\mu^{1/4})}\sum_{r > 0} e^{-\Omega(r / \sqrt{\mu})}  \\
            & \le \mu^{3/4} \cdot e^{-\Omega(\mu^{1/4})} + e^{-\Omega(\mu^{1/4})} \cdot \frac{1}{1 - e^{\Omega(1 /\sqrt{\mu})}}  \tag{since $\sum t^r = 1 / (1-t)$ for $t < 1$} \\
            & \le \mu^{3/4} \cdot e^{-\Omega(\mu^{1/4})} + e^{-\Omega(\mu^{1/4})} \cdot \Theta(\sqrt{\mu})  \tag{since $1 - e^{1/x} = \Theta(1/x)$ for $x \ge 1$}\\
            & \le e^{-\Omega(\mu^{1/4}) + O(\log \mu)} \\
            & \le e^{-\Omega(\mu^{1/4})}. 
\end{align*}    
It follows that
$$\E\left[\sum_{i = 1}^n A_i\right] \le n e^{-\Omega(\mu^{1/4})}.$$

We complete the proof by considering two parameter regimes, depending on whether $\mu \ge \log^{8} m$. If $\mu \ge \log^{8} m$, then
$$\E\left[\sum_{i = 1}^n A_i\right] \le n e^{-\Omega(\log^2 m)} \le m^{-\omega(1)},$$
where the final inequality uses the fact that $m \ge n$. By Markov's inequality, it follows that $\Pr[\sum A_i \ge 1] \le m^{-\omega(1)}$. Since $\sum_{i = 1}^n A_i$ is either $0$ or at least $1$, it follows that $\sum_{i = 1}^n A_i = 0 \le n / \mu^{\omega(1)}$ with high probability in $m$.

On the other hand, if $\mu \le \log^{8} m$, then we can complete the proof with McDiarmid's inequality. Define $Y_1, \ldots, Y_m$ to be the independent random bin choices for each of the balls $1, \ldots, m$, and define $F(Y_1, \ldots, Y_m) = \sum_{i = 1}^n A_i$. Since $F$ is a $\mu^{3/4}$-Lipschitz function determined by $m$ independent random variables, we have by McDiarmid's inequality (Theorem \ref{thm:mcdiarmid}) that 
$$\Pr[|F - \E[F]| \ge k \mu^{3/4} \sqrt{m}] \le e^{-\Omega(k^2)}.$$
Recalling that $\mu\le \polylog n$ and $m \le n \polylog n$, it follows with high probability in $m$ that
$$|F - \E[F]| \le \tilde{O}(\sqrt{n}).$$
This implies that
$$\sum_{i = 1}^n A_i \le n e^{-\Omega(\mu^{1/4})} + \tilde{O}(\sqrt{n}) \le n / \mu^{\omega(1)} + \tilde{O}(\sqrt{n}).$$
Again using the fact that $\mu \le \log^8 m$, this is at most $n / \mu^{\omega(1)}$. 
\end{proof}

Next we prove a basic lemma relating the error $\err{t}$ to the number of balls thrown in the spreading stage of round $t$.
\begin{lemma}
For $t \in \{0, \ldots, T\}$, if $\err{t} = O(n)$, then the number of balls that get thrown in the spreading stage of round $t$ is $m_{t} \pm O(n)$.  
\label{lem:Ettomt}
\end{lemma}
\begin{proof}
This holds trivially for $t = 0$, since the number of balls thrown in the spreading stage of round $0$ is exactly $m = m_0$. Fix a $t \in \{1, \ldots, T\}$, and let $w = \sum_{i \in [n]} \loadi{t}{i}$ be the number of balls in bins immediately prior to the spreading stage of round $t$. We have 
\begin{align*}
    w &=
    \sum_{i \in [n]} \loadi{t}{i} \\
    &= \sum_{i \in [n]} \loadi{t}{i} - \threshold{t} + \threshold{t}\\
    &= n\threshold{t} + \sum_{i \mid \loadi{t}{i} > \threshold{t}} \loadi{t}{i} - \threshold{t} + \sum_{i \mid \loadi{t}{i} < \threshold{t}}  -(\threshold{t} - \loadi{t}{i}) \\
    &= n\threshold{t} + \errover{t} - \errunder{t}.
\end{align*}

The number of balls thrown in the spreading stage of round $t$ is therefore
$$m - w = m - n \threshold{t} - \errover{t} + \errunder{t} = m_{t} - \errover{t} + \errunder{t} = m_{t} \pm O(n).$$
\end{proof}

\paragraph{Bounding the underfilled badness.} The next lemma allows us to bound how much $\errunder{t}$ grows when we increment $t$.
\begin{lemma}
Let $t \in \{0, \ldots, T - 1\}$. With high probability in $m_{t}$, if the number of balls thrown in the spreading stage of round $t$ is $m_{t} \pm O(n)$, then $\errunder{t + 1} \le \errunder{t} + n / \poly(\mu_{t})$. 
\label{lem:Eminus}
\end{lemma}
\begin{proof}
Focus on a fixed round $t \in \{0, \ldots, T - 1\}$. For each bin $i$, define $q_i$ so that $\loadi{t}{i} = \threshold{t} - q_i$. (We can think of $q_i$ as simply being a signed version of the round-$t$ error for bin $i$.) Let $k_i$ be the number of balls that the spreading stage in round $t$ places into bin $i$. Then we can express $\errunder{t + 1}$ as 
\begin{align*}
    \errunder{t + 1} & = \sum_i \max(\threshold{t+1} - (\threshold{t} - q_i + k_i), 0) \\
    & = \sum_i \max(\mu_{t} - \mu_{t}^{3/4} + q_i - k_i, 0) \tag{since $\threshold{t+1} - \threshold{t} = \mu_{t} - \mu_{t}^{3/4}$}\\
    & \le \sum_i \max(q_i, 0) + \sum_i \max(\mu_{t} - \mu_{t}^{3/4} - k_i, 0) \\
    & = \errunder{t} + \sum_i \max(\mu_{t} - \mu_{t}^{3/4} - k_i, 0).
\end{align*}
Finally, by Lemma \ref{lem:imbalance}, we have with high probability in $m_{t}$ that
$$\sum_i \max(\mu_{t} - \mu_{t}^{3/4} - k_i, 0) \le n / \poly(\mu_{t}),$$
which completes the proof. 
\end{proof}

\paragraph{Bounding the overfilled badness.} Next we bound the growth of $\errover{t}$ as $t$ increases.  
\begin{lemma}
Let $t \in \{0, \ldots, T - 1\}$. With high probability in $m_{t}$, if the number of balls thrown in the spreading stage of round $t$ is $m_t \pm O(n)$, then $\errover{t + 1} \le \errover{t} + n / \poly(\mu_t)$. 
\label{lem:Eplus}
\end{lemma}
To understand $\errover{t+1}$, it is helpful to consider the contribution of each bin $i$ to $\errover{t}$. Recall that bin $i$ contributes to $\errover{t}$ if it is overfilled bad in round $t$. In order for bin $i$ to be overfilled bad, it must have come across at least one of two failure modes: (1) The total number of round-$t$ balls in the bin is too small or (2) the total number of balls in the bin is too large. Before proving Lemma~\ref{lem:Eplus}, we show that the accumulation of these failure events across all bins is small with high probability. As we will see, we can prove this for both failure events via straightforward applications of Lemma~\ref{lem:imbalance}.

We begin with the first failure event:
\begin{clm}\label{clm:not-enough-round-t}
    Let $c_i \coloneqq |\colorSet[t+1][i]|$ be the number of balls in bin $i$ with round assignment $t+1$. Let $b_i \coloneqq \max(2 \mu_t^{3/4} - c_i, 0)$.
    \begin{equation}\label{eq:bi}
        \sum_{i = 1}^n b_i \le n / \poly(\mu_t)
    \end{equation}
    with high probability in $m_t$. 
\end{clm}
\begin{proof}
    Note that \eqref{eq:bi} is trivial for $\mu_t = O(1)$, so we will assume without loss of generality that $\mu_t$ is at least a large positive constant. 
    
    To establish \eqref{eq:bi}, let us first consider the \emph{total number} $c^{(t+1)}$ of balls with round assignment $t+1$ in the system. Since each ball has a $\Theta(m_{t} / m)$ probability of having round assignment $t+1$, the expected value $\E[c^{(t+1)}]$ is $\Theta(m_{t})$. By a Chernoff bound, it follows with high probability in $m_t$ that $c^{(t+1)} = \Theta(m_t)$.
    
    Conditioning on $c^{(t+1)} = \Theta(m_t)$, and applying Lemma \ref{lem:imbalance}, we have with high probability in $c^{(t+1)}$ (and thus with high probability in $m_t$) that 
    $$\sum_{i = 1}^n \max(c^{(t+1)} / n - c_i - (c^{(t+1)}/n)^{3/4}, 0) \le n / \poly(c^{(t+1)}/n) \le n / \poly(\mu_t).$$
    Since $c^{(t+1)}/n = \Theta(\mu_{t})$, and since $\mu_{t}$ is at least a sufficiently large positive constant, we have that\\ $c^{(t+1)} / n - c_i - (c^{(t+1)}/n)^{3/4} \ge 2\mu_t^{3/4} - c_i$. It follows that
    $$\sum_{i=1}^n b_i = \sum_{i = 1}^n \max(2 \mu_t^{3/4} - c_i, 0) \le n / \poly(\mu_t),$$
    and therefore that \eqref{eq:bi} holds.
\end{proof}
The second failure event follows by a straightforward application of Lemma~\ref{lem:imbalance}:
\begin{clm}\label{clm:overfill-is-rare}
    Suppose that the total number of balls thrown in round $t$ is $m_t \pm O(n)$. Let $k_i$ be the number of balls that the spreading stage in round $t$ places into bin $i$, and let $a_i = \max(k_i - (\mu_t + \mu_t^{3/4}), 0)$. Then, with high probability in $m_t$,
    $$\sum_{i=1}^n a_i \le n/\poly(\mu_t).$$
\end{clm}
\begin{proof}
    By Lemma~\ref{lem:imbalance}, 
    \begin{equation}\label{eq:ai}
        \sum_{i = 1}^n a_i = \sum_{i = 1} \max(k_i - (\mu_t + \mu_t^{3/4}), 0) \le n / \poly(\mu_t)
    \end{equation}
    with high probability in $m_t$.
\end{proof}

We can now prove Lemma~\ref{lem:Eplus}.
\begin{proof}[Proof of Lemma~\ref{lem:Eplus}]
Focus on a bin $i$ in round $t$, and define $q_i$ so that $\loadi{t}{i} = \threshold{t} + q_i$. Let $k_i$ be the number of balls that the spreading stage in round $t$ places into the bin $i$. After the spreading stage in round $t$, bin $i$ has $\threshold{t} + q_i + k_i$ balls. During the slice step in round $t + 1$, the bin would like to evict $\max(\threshold{t} + q_i + k_i - \threshold{t+1}, 0)$ balls with round assignment $t+1$. So, if the bin has $c_i \coloneqq |\colorSet[t+1][i]|$ balls with round assignment $t+1$, its overfill error in round $t + 1$ (i.e., its contribution to $\errover{t + 1}$) will be 
$$\max(\max(\threshold{t} + q_i + k_i - \threshold{t+1},0) - c_i, 0) = \max(\threshold{t} + q_i + k_i - \threshold{t+1} - c_i, 0).$$
Since $\threshold{t+1} - \threshold{t} = \mu_t - \mu_t^{3/4}$, our bound on the bin's overfill error is at most
$$\max(q_i, 0) + \max(k_i - c_i - \mu_t + \mu_t^{3/4}, 0).$$
Note that if $k_i \le \mu_t + \mu_t^{3/4}$ and $c_i \ge 2\mu_t^{3/4}$, then $\max(k_i - c_i - \mu_t + \mu_t^{3/4}, 0) = 0$. Thus we will bound $\max(k_i - c_i - \mu_t + \mu_t^{3/4}, 0)$ by bounding the amount by which $k_i$ exceeds $\mu_t + \mu_t^{3/4}$ or by which $c_i$ goes below $2\mu_t^{3/4}$. Specifically, if we define $a_i \coloneqq \max(k_i - (\mu_t + \mu_t^{3/4}), 0)$ and $b_i \coloneqq \max(2 \mu_t^{3/4} - c_i, 0)$, then we have
$$\max(q_i, 0) + \max(k_i - c_i - \mu_t + \mu_t^{3/4}, 0) \le \max(q_i, 0) + a_i + b_i.$$
This bounds each bin $i$'s contribution to $\errover{t + 1}$ by $\max(q_i, 0) + a_i + b_i$. Summing over the bins gives
$$\errover{t + 1} \le \sum_{i = 1}^n (\max(q_i, 0) + a_i + b_i) = \errover{t} + \sum_{i = 1}^n a_i + \sum_{i = 1}^n b_i.$$ By Claim~\ref{clm:overfill-is-rare} and Claim~\ref{clm:not-enough-round-t}, each of $\sum_{i = 1}^n a_i$ and $\sum_{i = 1}^n b_i$ are of the form $n / \poly(\mu_t)$ with high probability in $m_t$, completing the proof.

\end{proof}

Putting the previous lemmas together, we can now obtain a good overall bound on $\err{t}$.
\begin{lemma}
With high probability in $m_t$, we have
$$\err{t} \le n / \poly(\mu_t).$$
\label{lem:Etbound}
\end{lemma}
\begin{proof}
By Lemmas \ref{lem:Ettomt}, \ref{lem:Eminus}, and \ref{lem:Eplus} for each $t \in \{0, \ldots, T - 1\}$, we have that, with high probability in $m_t$, if $\err{t} = O(n)$, then 
$$\err{t + 1} \le \err{t} + n / \poly(\mu_t).$$
Since $\err{0} = 0 = O(n)$, it follows by induction on $t$ for every $t \in \{0, \ldots,  T - 1\}$ that, with probability at least 
$$1 - \sum_{q = 0}^{t} 1/ \poly(\mu_q),$$
we have
$$\err{t + 1} \le \sum_{q = 0}^{t} n / \poly(\mu_q) \le O(n).$$
Since $\sum_{q = 0}^{t} 1 / \poly(m_q) = 1 / \poly(m_t)$ and $\sum_{q = 0}^{t} n / \poly(\mu_q) = n / \poly(\mu_t)$, the lemma follows. 
\end{proof}

Finally, we can establish a bound on the total number of balls above height $\mu = m/n$ at the end of the construction:

\propfoo*
\begin{proof}
Applying Lemma \ref{lem:Etbound} to $t = T$, we have with high probability in $n$ that $\err{T} \le O(n)$. Since $\err{T} \le O(n)$, we know from Lemma \ref{lem:Ettomt} that the number $K$ of balls thrown in the spreading stage of round $T$ is at most $m_T = O(n)$. The total number of balls above height $\threshold{T} \le m/n$ at the end of the algorithm is therefore at most $\err{T + 1} + K \le O(n)$. 
\end{proof}

\subsection{The Recourse of Slice and Spread}\label{sec:slicespread}
In this section, we bound the expected recourse of Slice and Spread. To do so, we introduce Proposition~\ref{prop:recourse} which bounds the recourse for each $(\setBalls,\setBallsPrime,\hashOne,\hashTwo,\random)$:
\begin{restatable}{prop}{proprecourse}\label{prop:recourse}
    Fix any neighboring datasets $\setBalls$ and $\setBallsPrime$, and $h_1$, $h_2$, and $R$. The recourse of $(\setBalls,\setBallsPrime, h_1, h_2, R)$ incurred by Slice and Spread is $O(\log\log (m/n))$.
\end{restatable}

To prove Proposition~\ref{prop:recourse}, we will make use of some additional notation and definitions. 

As a convention, let $|\setBallsPrime| > |\setBalls|$. Let the \textbf{\emph{extra ball}} $x^*$ be the ball that appears in $\setBallsPrime$ but not $\setBalls$. Let \textbf{\emph{world $0$}} be the scenario that considers $\setBalls$, and let \textbf{\emph{world $1$}} be the scenario that considers $\setBallsPrime$. In general, to distinguish between the two worlds, we will add the notation of $(w)$ to the end of the variable, where $w$ is the variable for world. For example, $\load{t}{i}[w]$, $\colorSet[t][i][w]$, and $\rethrow{t}[i][w]$. Let $i^* \coloneqq \hashOne[x^*]$ be the first-choice bin of the extra ball. If $x^*$ has a round assignment, then let $t^* \coloneqq \cmap[x^*]$ be the round assignment. Otherwise, let $t^* = 0$.

Much of the analysis will be spent discussing the differences between the states of the data structure between world $0$ and world $1$. Let $d^{(t)}_i \coloneqq |\ell^{(t)}_i(0) - \ell^{(t)}_i(1)|$ be the \textbf{\emph{discrepancy of bin $i$}}. Let $d^{(t)} \coloneqq \sum_i d^{(t)}_i$ be the \textbf{\emph{total discrepancy}} across the two worlds in round $t$. 

Finally, call a ball $x$ a \textbf{\emph{red ball}} if $x$ was evicted from a bin in some world $w$ but not simultaneously in $\bar{w}$. Let $E^{(t)}_i$ be the number of red balls evicted from bin $i$ during round $t$ in world $0$ and world $1$. That is, $E_i^{(t)} = |\rethrow{t}[i][0] \Delta \rethrow{t}[i][1]|$. 

With these conventions in mind, we can now begin the analysis. We start with a few simple lemmas that will help us reason about the recourse. The first lemma says that the number of balls $x \in \colorSet[t][i][0]$ that get assigned to different bins in worlds $0$ and $1$ is at most $E_i^{(t)}$.
\begin{lemma}\label{lem:recourse-to-red}
    Let $\phi_0$ and  $\phi_1$ represent the final allocations of the balls to bins in world $0$ and world $1$, respectively. For every $(t,i) \in [T]\times [n]$,
    \begin{equation}\label{eq:bounding-recourse-round}
        \sum_{x \in \colorSet[t][i][0]} \ind\Paren{\phi_0(x) \ne \phi_1(x)} \le E^{(t)}_i.
    \end{equation}
\end{lemma}
\begin{proof}
    For a given ball $x \in \colorSet[t][i][0]$, and a given world $w \in \{0, 1\}$, the final bin $\phi_w(x)$ that ball $x$ gets assigned to in world $w$ is 
    $$\begin{cases} h_1(x) & \text{ if } x \not\in \rethrow{t}[i][w] \\ h_2(x) & \text{ if }x \in \rethrow{t}[i][w]\end{cases}.$$ 
    It follows that 
    $$\ind\Paren{\phi_0(x) \ne \phi_1(x)} = \ind\Paren{x \in \rethrow{t}[i][0] \triangle \rethrow{t}[i][1]},$$
    which implies that 
    \begin{equation}\label{eq:googoo}
        \sum_{x \in \colorSet[t][i][0]} \ind\Paren{\phi_0(x) \ne \phi_1(x)} = \sum_{x \in \colorSet[t][i][0]} \ind\Paren{x \in \rethrow{t}[i][0] \triangle \rethrow{t}[i][1]}  \le |\rethrow{t}[i][0] \triangle \rethrow{t}[i][1]| =  E^{(t)}_i.
    \end{equation}
    We note that the inequality in \eqref{eq:googoo} is due to the fact that the summation is only over balls shared in common between world $0$ and world $1$ (which does not include $x^*$), but $E_i^{(t)}$ may contain $x^*$.
\end{proof}

This next lemma says that, if we want to bound $E_i^{(t)}$, it suffices to bound $d_i^{(t)}$.
\begin{lemma}\label{lem:red-balls}
    Let $(t,i) \in [T]\times [n]$ such that $(t,i) \ne (t^*,i^*)$. Then, 
    $$E_i^{(t)} \le d_i^{(t)}.$$
\end{lemma}

This lemma makes use of the following simple claim:
\begin{restatable}{clm}{clmwhenasubset}\label{clm:when-a-subsest}
    Let $(t, i) \neq (t^*, i^*)$, and let $w \in \{0, 1\}$ be the world $w$ for which $\load{t}{i}(w) = \load{t}{i}(\bar{w}) + d_i^{(t)}$. Then 
    \begin{equation}\label{eq:rethrowcontain}
        \rethrow{t}[i][\bar{w}] \subseteq \rethrow{t}[i][w].
    \end{equation}
\end{restatable}
\begin{proof}
        Since $(t,i) \ne (t^*,i^*)$, the sets $\colorSet[t][i][0]$ and $\colorSet[t][i][1]$ are equal and furthermore have the same ordering as determined by the canonical ordering. Since $\rethrow{t}[i][w]$ and $\rethrow{t}[i][\bar{w}]$ are prefixes of the same ordered set, and since $\load{t}{i}(w) \ge \load{t}{i}(\bar{w})$, it follows that $\rethrow{t}[i][\bar{w}] \subseteq \rethrow{t}[i][w]$.
\end{proof}

\begin{proof}[Proof of Lemma~\ref{lem:red-balls}]
By Claim~\ref{clm:when-a-subsest}, 
$$E_i^{(t)} = ||\rethrow{t}[i][0]| - |\rethrow{t}[i][1]||.$$
Recall that, in each world $w \in \{0, 1\}$, the number $|\rethrow{t}[i][w]|$ of balls that we evict from bin $i$ in round $t$ is $\min(|\colorSet[t][i](w)|,\load{t}{i}(w) - \threshold{t})$. It follows that 
$$||\rethrow{t}[i][w]| - |\rethrow{t}[i][\bar{w}]|| = |\min(|\colorSet[t][i](0)|,\load{t}{i}(0) - \threshold{t}) - \min(|\colorSet[t][i](1)|,\load{t}{i}(1) - \threshold{t})|.$$
Since $(t, i) \neq (t^*, i^*)$, we have $\colorSet[t][i](0) = \colorSet[t][i](1)$. Using the identity $|\min(a, b) - \min(a, c)| \le |b - c|$, we can conclude that 
$$|\min(|\colorSet[t][i](0)|,\load{t}{i}(0) - \threshold{t}) - \min(|\colorSet[t][i](1)|,\load{t}{i}(1) - \threshold{t})| \le |\load{t}{i}(0) - \load{t}{i}(1)| = d_i^{(t)}.$$
\end{proof}

We now get to the meat of the section, where, in the next two lemmas, we show how to bound the discrepancy $d^{(t)}$ at the beginning of each round $t$. 
\begin{lemma}\label{lem:bounding-discrepancy}
    For all $t \ne t^*$, $d^{(t+1)} \le d^{(t)}$.
\end{lemma}
\begin{proof}
    We can decompose any round $t$ into the slicing stage and the spreading stage. Recall that $E_i^{(t)}$ is the total number of balls that bin $i$ evicts in round $t$ in one world but not the other. Let $R^{(t)}_i \coloneqq \sum_{i \in [n]} |\{x \in E^{(t)}_i \mid h_2(x) = i\}|$ be the number of red balls that bin $i$ receives in round $t$ (in either world) during the spreading stage.
    We begin by showing the following inequality:
    \begin{equation}
        d^{(t+1)}_i \le d^{(t)}_i - E^{(t)}_i + R^{(t)}_i.
        \label{eq:dtjle}
    \end{equation} 

    We first consider the effect of the slicing stage in round $t$ on the discrepancy of bin $i$. Since $t \neq t^*$, we know by Lemma \ref{lem:red-balls} and Claim \ref{clm:when-a-subsest} that there is some $w \in \{0, 1\}$ such that $\load{t}{i}[w] = \load{t}{i}[\bar{w}] + d^{(t)}_i \ge \load{t}{i}[\bar{w}] + E^{(t)}_i$ and such that, during the slicing stage of round $t$, $E^{(t)}_i$ more balls are sliced from bin $i$ in world $w$ than in world $\bar{w}$. Thus, the slicing stage in round $t$ decreases the the discrepancy of bin $i$ by $E^{(t)}_i$. 

    We now show that the spreading stage increases the discrepancy of bin $i$ by at most $R^{(t)}_i$. This is because, during the spreading stage, the only balls that can be placed into bin $i$ in one world but not the other are red balls. Since the number of red balls placed into bin $i$ (in either world) during the spreading stage of round $t$ is $R^{(t)}_i$, it follows that the discrepancy of bin $i$ increases by at most $R^{(t)}_i$ during the spreading stage. 
    
    Having established \eqref{eq:dtjle}, we can sum across the bins to get 
    \begin{align*}
        d^{(t+1)} &= \sum_i d^{(t+1)}_i \\
        &\le \sum_i \left(d^{(t)}_i - E^{(t)}_i + R^{(t)}_i\right)\\
        &= d^{(t)} - \sum_i E^{(t)}_i + \sum_i R^{(t)}_i \\
        &= d^{(t)}. 
    \end{align*}
\end{proof}
\begin{lemma}\label{lem:special-red-balls}
    Suppose that $t^* \ne 0$. We have $E_{i^*}^{(t^*)} \le d_{i^*}^{(t^*)} + 2$. Furthermore, if $t^* < T$, then $d^{(t^* + 1)} \le d^{(t^*)} + 2$.
\end{lemma}
\begin{proof}
    To prove this lemma, let us also introduce a hybrid world $2$, in which everything is the same as world $1$ except the extra ball is not assigned a round (i.e. $t^* = 0$).

    At the beginning of round $t^*$, worlds 1 and 2 share the same exact state (but with a difference in round assignment for the extra ball). During the slicing step in round $t^*$, all of the bins except for bin $i^*$ evict the same set of balls as each other in worlds 1 and 2. The only bin that behaves differently in the two worlds is bin $i^*$, which may evict up to one ball in each of worlds 1 and 2 that it doesn't evict in the other (call this the \defi{World-Difference Observation}).
    
    We now argue that $E_{i^*}^{(t^*)} \le d_{i^*}^{(t^*)} + 2$. Define $E_{0, 1}, E_{1, 2}, E_{0, 2}$ so that $E_{w, w'}$ is the number of balls $x$ such that $h_1(x) = i^*$, and such that $x$ is evicted in round $t^*$ in one of worlds $w, w'$ but not in the other (i.e., $x$ is a round-$t^*$ red ball when comparing worlds $w$ and $w'$). 
    
    By the triangle inequality, we have that $E_{0,1} \le E_{0,2} + E_{1,2}$. By Lemma~\ref{lem:red-balls}, we know that $E_{0,2} \le d^{(t^*)}_{i^*}$. Putting these together gives $E_{0,1} \le d^{(t^*)}_{i^*} + E_{1,2}$. Finally, we have by the World Difference Observation that, during round $t^*$, worlds $1$ and $2$ evict identical sets of balls except for up to one ball from bin $i^*$ in each world. It follows that $E_{1, 2} \le 2$, which implies that $E_{0, 1} \le d^{(t^*)}_{i^*} + 2$. 

    Next we argue that, if $t^* < T$, then $d^{(t^* + 1)} \le d^{(t^*)} + 2$. Define $D_{0, 1} = d^{(t^* + 1)}$ to be the discrepancy between worlds 0 and 1 at the beginning of round $t^* + 1$; define $D_{1, 2}$ to be the discrepancy between worlds 1 and 2 at the beginning of round $t^* + 1$; and define $D_{0, 2}$ to be the discrepancy between worlds $0$ and $2$ at the beginning of round $t^* + 1$. 

    By the triangle inequality, $D_{0, 1} \le D_{1, 2} + D_{0, 2}$. By Lemma \ref{lem:bounding-discrepancy}, we know that $D_{0, 2} \le d^{(t^*)}$. Therefore, to complete the proof that $d^{(t^* + 1)} \le d^{(t^*)} + 2$, it suffices to show that $D_{1, 2} \le 2$.  By the World Difference Observation, the states in worlds 1 and 2 at the beginning of round $t^{*} + 1$ differ in the positions of at most two balls. It follows that $D_{1, 2} \le 2$, as desired. 
\end{proof}

Finally, by putting the preceding lemmas together, we can prove Proposition \ref{prop:recourse}.
\proprecourse*
\begin{proof}
    We assume that $t^* \ne 0$. The proof follows by stringing together the previous lemmas as follows:
    \begin{align*}
        & 1 + \sum_{x \in \setBalls \cap \setBallsPrime} \ind\Paren{\phi_0(x) \ne \phi_1(x)}\\
        &= 1 + \sum_{(t,i)\in [T]\times[n]}\ind\Paren{\phi_0(x) \ne \phi_1(x)}\\
        &= 1 + \sum_{x \in \colorSet[t^*][i^*][0]}\ind\Paren{\phi_0(x) \ne \phi_1(x)} + \sum_{\substack{(t,i)\in [T]\times[n]\\(t,i) \ne (t^*,i^*)}}\sum_{x \in \colorSet[t][i][0]}\ind\Paren{\phi_0(x) \ne \phi_1(x)}\\
        &\le 1 + E_{i^*}^{(t^*)} + \sum_{\substack{(t,i)\in [T]\times[n]\\(t,i) \ne (t^*,i^*)}} E_i^{(t)} \tag{by Lemma~\ref{lem:recourse-to-red}}\\
        &\le 1 + d_{i^*}^{(t^*)} + 2 + \sum_{\substack{(t,i)\in [T]\times[n]\\(t,i) \ne (t^*,i^*)}} d_i^{(t)} \tag{by Lemma~\ref{lem:red-balls} and Lemma~\ref{lem:special-red-balls}}\\
        &= 3 + \sum_{t \in [T]} d^{(t)}\\
        &= 3 + \sum_{t \in [T]} (d^{(1)} + 2)  \tag{by Lemma~\ref{lem:bounding-discrepancy} and Lemma~\ref{lem:special-red-balls}} \\
        &= 3 + T \cdot 3 \\
        &= O(\log\log (m/n)).
    \end{align*}
    We can prove for the case of $t^* = 0$ similarly. In this case, note that we can reuse the same sequence of inequalities above, except where the $t^*$ terms disappear, since $(t^*,i^*) \notin [T]\times[n]$.
\end{proof}

We can now prove Theorem~\ref{thm:offline-slice-and-spread}
\thmofflinesliceandspread*
\begin{proof}
    Follows from Proposition~\ref{prop:recourse} and Proposition~\ref{prop:foo}.
\end{proof}
\section{Reducing the Overload to \texorpdfstring{$O(1)$}{O(1)}}\label{sec:reducing-overload}
In this section, we will show how to reduce the overload to $O(1)$, while still achieving expected recourse $O(\log \log (m/n))$. Our approach will be through a black-box construction, taking any algorithm with the properties from the previous section, and turning it into an algorithm with overload $O(1)$. Note that, throughout the section, we will continue with the WLOG assumption from Section~\ref{sec:preliminaries} that the number of balls is always either $m - 1$ or $m$.

Throughout the section, label every ball $x$ randomly one of three types: Type $1$, Type $2$, and Type $3$, where the probability $p_i$ of each Type $i$ is given by $p_1 = 0.89, p_2 = 0.1, p_3 = 0.01$.  

Consider a history-independent algorithm $\alg$ that is defined for sets of $m - 1$ and $m$ balls, and that assigns them to $n$ bins. Call $\alg$ a \defi{good pre-baking algorithm} if it has the following properties:
\begin{itemize}
    \item \textbf{First Choices for Types $1$ and $2$:}  For Type $1$ and Type $2$ balls $x$, $\alg$ assigns $x$ to $\hashOne[x]$ and does not evaluate $\hashTwo[x]$.
    \item \textbf{Small Cumulative Overload:} For a given set $\setBalls$ of $m$ balls, we have with high probability in $n$ that $\alg(S)$ has cumulative overload $O(n)$.
\end{itemize}

We will show that, given any good pre-baking algorithm, we can construct a new algorithm that has the same asymptotic expected recourse, but that has overload $O(1)$:
\begin{thm}
Given any good pre-baking algorithm $\alg$ with expected recourse $\recourse[\alg]$, we can construct a new history-independent algorithm with expected recourse $O(\recourse[\alg])$ and whose overload is $O(1)$ with high probability in $n$. 
\label{thm:postprocess}
\end{thm}

\subsection{Background: History Independent Cuckoo Hashing on \texorpdfstring{$O(n)$}{O(n)} Elements} One tool that we will need in this section is a simple but elegant paradigm due to Naor, Segev, and Wieder \cite{naor2008history} for how to construct an efficient history-independent cuckoo hash table. We will describe their paradigm in the context of our balls-and-bins problem, where $m = O(n)$. 

\paragraph{Orienting each component separately.} Given a set $B$ of balls, define the \defi{cuckoo graph} $G(B)$ to be the graph with vertices in $[n]$ and edges $\{(\hashOne[x], \hashTwo[x]) \mid x \in B\}$. We can think of an assignment of balls to bins as an orientation on the graph $G$, where the edge corresponding to a ball $x$ points to vertex $h_i(x)$ iff $x$ is assigned to bin $h_i(x)$. The maximum load across the bins corresponds to the maximum in-degree in the directed graph. 

Let $\orient$ be an arbitrary deterministic algorithm that, given a connected component $C$ of a graph $G$, produces a minimum-in-degree orientation of $C$. Define the \defi{canonical orientation} $\orient[G(B)]$ to be the orientation obtained by running $\orient$ on each connected component of $G(B)$. We will also define the \defi{Canonical Orientation Procedure} as the procedure for running $\orient$ on each connected component of $G(B)$.

Naor, Segev, and Wieder \cite{naor2008history} make two key observations:
\begin{enumerate}
    \item Fix a ball $x \in B$, and define $C$ to be the connected component in $G(B)$ containing edge $(\hashOne[x], \hashTwo[x])$. Let $|C|$ denote the number of edges in $C$. If we look at $\orient[G(B)]$, the number of edges whose orientations change when we remove $x$ is at most $|C|$. That is,
    $$|\orient[G(B)] \triangle \orient[G(B \setminus \{x\})]| \le |C|.$$
    \item If every connected component in $G(B)$ contains $O(1)$ cycles, then $\orient[G(B)]$ has maximum in-degree $O(1)$; and, thus, the corresponding balls-to-bins assignment has maximum load $O(1)$.
\end{enumerate}

Using these observations, Naor, Segev, and Wieder \cite{naor2008history} showed that, for sets $B$ of $n/2 - \Omega(n)$ balls, the canonical orientation $\orient[G(B)]$ achieves both expected recourse $O(1)$ and overload $O(1)$ (with high probability). This gave them a simple way of constructing efficient history-independent cuckoo hashing. 

\paragraph{Generalizing to our setting.} In this section, we will make repeated use of the Canonical Orientation Procedure. The main difficulty that we will need to overcome is that, a priori, the balls $B$ to which we wish to apply the Canonical Orientation Procedure do not necessarily produce a graph $G$ with small components (or even with few cycles per component). We will overcome this challenge through a mixture of algorithmic techniques (i.e., ways of changing which set of balls we need to apply the procedure to) and analytical techniques (i.e., analyses of graphs $G$ whose edges may be neither independent nor uniformly random). 

\paragraph{Extending to slightly larger sets of balls.} Before continuing, we remark that there is a simple way to extend the Canonical Orientation Procedure to allow for $m = O(n)$ balls, rather than requiring $m \le n / 2 - \Omega(n)$. We can simply (1) partition the balls at random into $O(1)$ subsets, each of size $n/2 - \Omega(n)$; and then (2) process each subset separately using the Canonical Orientation Procedure.

We refer to this procedure as the \defi{Extended Canonical Orientation (ECO) Procedure}. Formally, for a set of balls $B$, we define the extended orientation $\orient^*(G(B))$ as follows. Let $p \in (0, 1)$ be a sufficiently small positive constant satisfying the constraint that $p^{-1}$ is an integer. Partition the graph $G(B)$ into $p^{-1}$ random subgraphs $G_1(B), G_2(B), \ldots, G_{p^{-1}}(B)$, where each edge in $G(B)$ gets assigned (via a hash function) to a random subgraph. Then the orientation $\orient^*(G(B))$ is obtained by computing the canonical orientation of each subgraph, that is,
$$\orient^*(G(B)) \coloneqq \bigcup_{j = 1}^{p^{-1}} \orient(G_j(B)).$$

As a convention throughout the section, we will use the notations $G_j(B)$, $\orient(G_j(B))$ and $\orient^*(G(B))$ whenever we discuss an application of the ECO Procedure to a set $B$ of balls. We will also use $p$ to refer to the small positive constant used within the procedure.  

\subsection{The Algorithm}
To present the algorithm that we will use to prove Theorem \ref{thm:postprocess}, we begin motivating the basic ideas which the algorithm will use.

\paragraph{Idea 1: Apply the ECO Procedure to balls above height $m/n$.}
Let $B$ denote the set of balls that $\alg$ places above height $m/n$ (ties for which balls are in $B$ can be broken arbitrarily). The first idea one might try is to rearrange the balls in $B$ (i.e., reselect which hash function each ball uses) so that they are in an arrangement with overload $O(1)$. Specifically, we could attempt to simply apply the ECO Procedure to $B$. 

The problem with this idea is that the balls $B$ have \emph{spoiled randomness}. Because $B$ is constructed by an algorithm $\alg$ that gets to see the hashes of each ball, we cannot assume that the hashes $\{h_1(x), h_2(x) \mid x \in B\}$ are independent or uniformly random. This means that the graphs $G_1(B), G_2(B), \ldots, G_{p^{-1}}(B)$ do not necessarily have small connected components (or, for that matter, a small number of cycles in each connected component).

\paragraph{Idea 2: Swapping the balls in $B$ for balls with fresh randomness.} To rectify the issue of spoiled randomness, we can attempt the following procedure for swapping the balls in $B$ with new balls $B'$ that have fully random hashes. For this, we use what we call the \defi{Two-Phase Swapping Procedure:}
\begin{itemize}
    \item \textbf{Phase 1:} For each bin $i$ with $a_i$ Type $2$ balls and $k_i$ balls from $B$, remove $\min(a_i, k_i)$ Type $2$ balls from the bin and add them to an intermediate set $\bint$ (when choosing which Type $2$ balls to place in $\bint$, we may use an arbitrary deterministic tie-breaking rule). If $a_i < k_i$, then say that $k_i - a_i$ of the balls from $B$ in the bin experience \defi{Phase 1 failures}, and we place $k_i - a_i$ \defi{dummy balls} into $\bint$, each of which has $\hashOne[x] = i$ and $\hashTwo[x]$ generated independently and uniformly at random. The dummy balls take their second hashes from a random tape $R_i$ for bin $i$. 
    \item \textbf{Phase 2:} Place each $x \in \bint$ into its second-choice hash $\hashTwo[x]$. Now, for each bin $i$ with $b_i$ Type $1$ balls and $j_i$ balls from $\bint$, remove $\min(b_i, j_i)$ Type $1$ balls from the bin and add them to a set $\bafter$ (when choosing which Type $1$ balls to place in $\bafter$, we may use an arbitrary deterministic tie-breaking rule). If $b_i < j_i$, then say that $j_i - b_i$ of the Type $2$ balls in the bin experience \defi{Phase 2 failures}, and place $j_i - b_i$ \defi{dummy balls} into $\bafter$, each of which has $\hashOne[x] = i$ and $\hashTwo[x]$ generated independently and uniformly at random. The dummy balls take their second hashes from a second random tape $R_i'$ for bin $i$. Finally, define $B' \coloneqq \bafter$. 
\end{itemize}

The good news about the above procedure is that the balls $x \in B'$ have mutually independent and fully random hashes $\hashOne[x], \hashTwo[x]$. This will allow us to, later in the section, successfully apply the ECO Procedure to the graph $G(B')$.

\begin{lemma}[Random Hashes in $B'$]
Fix some set $B$ of $O(n)$ balls. Condition on fixed values of $\hashOne[x]$ for each ball $x \in B$, and on fixed values of $\hashTwo[x]$ for each Type 3 ball in $B$. Even with these conditions, the balls $x \in B' = \bafter$ (including the dummy balls) have mutually independent and fully random hashes $\hashOne[x], \hashTwo[x]$. 
\label{lem:randomG}
\end{lemma}
\begin{proof}
The balls $x \in \bint$ take their first hashes $\hashOne[x]$ from the first hashes of balls in $B$ (these hashes are spoiled), but have second hashes $\hashTwo[x]$ that are fully random. The balls $x \in \bafter$ take their first hashes $\hashOne[x]$ from the second hashes of balls in $\bint$ (these hashes are already random!) and have second hashes $\hashTwo[x]$ that are fully random. The result is that the balls $x \in \bafter$ have mutually independent and fully random hashes $\hashOne[x], \hashTwo[x]$. 
\end{proof}

The other good news is that there will likely not be too many balls that incur Phase 1 or Phase 2 failures. In fact, if $m/n = \omega(\log n)$, one can argue with high probability in $n$ that there are \emph{no} Phase 1 or Phase 2 failures. 

On the other hand, if $m/n = O(\log n)$, then there may be many balls that incur Phase 1 or Phase 2 failures. These balls may still sit above height $m/n$, and could cause a large overload. 

\paragraph{Idea 3: Over-approximating Phase 1 and Phase 2 failures with a well-behaved set.}
Let $\failSet$ denote the set of balls that incur Phase 1 and Phase 2 failures (within each bin, we may break ties arbitrarily when deciding which balls are considered to have experienced a failure). 

One challenge in reasoning about $\failSet$ is that the distribution of the hashes of the balls in $\failSet$ may be influenced by the behavior of the algorithm $\alg$. A critical step in this section is to define a larger set $\failSuperset$ that does not depend on the algorithm $\alg$, and that can be used as a proxy for $\failSet$. 
% \rose{if time, expand as definitions so it is readable.} 
Let $\mu = m/n$, let $\epsilon = 0.0001$, and recall that each ball has probability $p_j$ of being Type $j$, where $p_1 = 0.89, p_2 = 0.1, p_3 = 0.01$. We say that a bin $i$ is \defi{Type $j$ overloaded from hash $h_k$} if the number of balls $x$ satisfying $h_k(x) = i$ is at least $(p_j + \epsilon)\mu$. Likewise, we say that a bin $i$ is \defi{Type $j$ underloaded from hash $h_k$} if the number of balls $x$ satisfying $h_k(x) = i$ is at most $(p_j - \epsilon) \mu$. A bin $i$ is said to be \defi{Type $j$ overloaded} (resp.~\defi{Type $j$ underloaded}) if it is Type-$j$ overloaded (resp.~Type $j$ underloaded) from at least one of the hash functions $h_1, h_2$. 

With these definitions in mind, we construct $\failSuperset$ to consist of all balls $x$ such that:
\begin{itemize}
    \item $x$ is a Type $i$ ball for some $i$, and there is some $j \in \{1, 2\}$ such that $h_j(x)$ is a Type $i$ overloaded bin from $h_j$. 
    \item $x$ is a Type $i$ ball for some $i$, and at least one of $h_1(x), h_2(x)$ is a Type $j$ underloaded bin for some $j < i$. 
\end{itemize}

Say that the system is \defi{$\failSuperset$-safe} if removing the balls $\failSuperset$ from the system results in every bin having load at most $m/n$. A key step in the analysis will be to argue (in Lemma \ref{lem:Fsafe}) that, after we apply the Two-Phase Swapping Procedure, the state of the system is $\failSuperset$-safe.  (Note that, at this point in time, the balls $B'$ are also not in the system.) This captures the way in which $\failSuperset$ is a good ``proxy'' for $\failSet$. 

\noindent For describing our algorithm, it will be helpful to think of $\failSuperset$ as a union of sets defined as follows:
\begin{itemize}
\item For each $a, b \in [3]$ with $a > b$, and for each $c, d \in [2]$, define $X^{(a, b, c, d)}$ to be the set of Type $a$ balls $x$ for which $h_c(x)$ is a bin that is Type $b$ underloaded due to $h_d$. 
\item For each $a \in [3]$ and $c \in [2]$, define $Y^{(a,c)}$ to be the set of Type $a$ balls $x$ for which $h_c(x)$ is a bin that is Type $a$ overloaded due to $h_c$. 
\item Then $\failSuperset$ is the union of the $X^{(a, b, c, d)}$s and $Y^{(a, c)}$s.
\end{itemize}

The final step in our algorithm will be to apply the ECO Procedure to each $X^{(a, b, c, d)}$ and to each $Y^{(a, c)}$. This will be somewhat tricky because the edges in a given $X^{(a, b, c, d)}$ (or $Y^{(a, c)}$) are not independent or uniformly random. Nevertheless, the edges in a given $X^{(a, b, c, d)}$ or $Y^{(a, c)}$ will turn out to have a nice enough combinatorial structure that we will be able to succeed in applying the ECO Procedure to them. 

\paragraph{The full algorithm.}
We can now describe the full algorithm for proving Theorem \ref{thm:postprocess}. Let $p \in (0, 1)$ be a sufficiently small positive constant to be used in the ECO Procedure.
\begin{enumerate}
    \item First, apply $\alg$ to get an initial orientation $\alg(S)$ of the balls.
    \item Apply the Two-Phase Swapping Procedure to $\alg(B)$, and let $P_1, P_2, B'$ be as defined in the procedure. 
    \item Next, apply the ECO Procedure to the balls in $B'(S)$. 
    \item Finally, define $\failSuperset = (\bigcup_{a, b, c, d} X^{(a, b, c, d)}) \cup (\bigcup_{a, c} Y^{(a, c)})$ as above, and reorient the edges in $F'$ as follows: For each $X^{(a, b, c, d)}$ and for each $Y^{(a, c)}$, one after another, apply the ECO Procedure to the graph induced by those balls. If a ball is in multiple such sets, it ends up using the orientation given by the final of the sets that the ball is in.
\end{enumerate}

Note that $P_1$, $P_2$, $B'$, $F'$, $X^{(a, b, c, d)}$, and $Y^{(a, c)}$ are all implicit functions of $S$ and will sometimes also be written with $S$ as an argument (e.g., $P_1(S)$). 

\subsection{Bounding Overload}\label{sec:postprocessoverload}

To prove Theorem \ref{thm:postprocess}, the first step is to bound the overload of our final algorithm by $O(1)$ (with high probability). 

To do this, we begin by analyzing the set $\failSuperset$. Recall that the system is said to be $\failSuperset$-safe if removing the balls $\failSuperset$ from the system results in every bin having load at most $m/n$. The next lemma argues that, after the Two-Phase Swapping Procedure, the system is $\failSuperset$-safe. 
\begin{lemma}
After applying the Two-Phase Swapping Procedure to $\alg(B)$, the system is $\failSuperset$-safe. (Note that, at this point in time, the balls $B'$ are also not in the system.)
\label{lem:Fsafe}
\end{lemma}
\begin{proof}
We begin by arguing the following claim:
\begin{clm}\label{clm:fsafe-phase1}
     The system is $\failSuperset$-safe after Phase 1. (Note that, at this point, the balls $\bint(\setBalls)$ have been removed from the system).
\end{clm}
To rephrase Claim~\ref{clm:fsafe-phase1}, we would like to show that, after removing the balls in $P_1$ from the system, and then removing the balls in $F'\setminus P_1$, each bin has load at most $\mu = m/n$. To prove this, it is helpful to think of the equivalent process of first removing $F'$ and then removing $P_1 \setminus F'$, where $F'$ and $P_1$ are known in advance.

Before proving the claim, we introduce some notation. We fix a bin $i$ and define the variables $x_i$, $y_i$, $z_i$, respectively, to be the number of Type 3, Type 2, and Type 1 balls in the bin after the $F'$ balls are removed but before the $P_1 \setminus F'$ balls are removed:
\begin{itemize}
    \item $x_i$ is the number of non-$F'$ Type $3$ balls in bin $i$ before Phase 1.
    \item $y_i$ is the number of non-$F'$ Type $2$ balls in bin $i$ before Phase 1.
    \item $z_i$ is the number of non-$F'$ Type $1$ balls in bin $i$ before Phase 1.
\end{itemize}
Let $o_i$ be the number of balls in bin $i$ above height $\mu$ after removing $F'$ from the system, i.e.
\begin{equation}\label{eq:oi}
    o_i = \max(0,x_i + y_i + z_i - \mu).
\end{equation}
\begin{proof}[Proof of Claim~\ref{clm:fsafe-phase1}]
We prove the claim for each bin $i$. We divide the proof of the claim into two cases based on $y_i$: (1) $y_i = 0$, and (2) $y_i \ne 0$.

\paragraph{Case $1$:} It suffices to show that $o_i = 0$. To see why $o_i = 0$ in Case $1$, it is helpful to note that 
\begin{equation}x_i \le 2(p_3 + \epsilon)\mu.
\label{eq:xi}
\end{equation}
and
\begin{equation}z_i \le (p_1 + \epsilon) \mu,
\label{eq:zi}
\end{equation}
where \eqref{eq:xi} comes from the fact that, for each $k \in \{1, 2\}$, $\failSuperset$ contains Type 3 balls $x$ for which $h_k(x)$ is in a Type 3 overloaded bin from $h_k$; and \eqref{eq:zi} comes from the fact that $\failSuperset$ contains all Type 1 balls $x$ for which $h_1(x)$ is a bin that is Type 1 overloaded from $h_1$. 

Since $y_i = 0$, it follows that $x_i + y_i + z_i \le (p_1 + 2p_3 + 2\epsilon) \mu < \mu$, which implies by \eqref{eq:oi} that $o_i = 0$.

\paragraph{Case $2$:} It's helpful to first establish two inequalities. First, since $\failSuperset$ contains every Type 2 ball $x$ for which $h_1(x)$ is a Type 2 overloaded bin from $h_1$, we know that
\begin{equation}\label{eq:yii}
    y_i \le (p_2 + \epsilon)\mu
\end{equation}
Combining \eqref{eq:yii} and \eqref{eq:zi} gives $y_i + z_i \le p_1\mu + p_2\mu + 2\epsilon \mu \le \mu$, which implies by \eqref{eq:oi} that
\begin{equation}\label{eq:oilexi}
    o_i \le x_i.
\end{equation}

We now subdivide Case $2$ into two more cases: (2a) $y_i < p_2\mu - \epsilon\mu$ and (2b) $y_i \ge p_2\mu - \epsilon\mu$. In Case $2a$, the fact that $y_i < p_2\mu - \epsilon\mu$ causes all balls of Type $3$ to be added to $F'$ and therefore removed from the bin, i.e. $x_i = 0$ (recall that $F'$ contains every Type 3 ball in a Type 2 overloaded bin). Thus, by \eqref{eq:oilexi}, $o_i = 0$, completing the proof in this case. In Case $2b$, the fact that $y_i \ge p_2\mu - \epsilon\mu$ combined with \eqref{eq:xi} implies $y_i \ge x_i$, which combined with \eqref{eq:oilexi} shows $y_i \ge o_i$. This means that at least $o_i$ Type $2$ balls from bin $i$ are removed from the bin, completing the proof.
\end{proof}

Using Claim~\ref{clm:fsafe-phase1}, we can now argue that each bin $i$ contains at most $\mu$ non-$\failSuperset$ balls at the end of Phase $2$.

Since, for each $k \in \{1, 2\}$ and $j \in \{1, 2, 3\}$, $\failSuperset$ contains all Type $j$ balls $x$ for which $h_k(x)$ is a bin that is Type $j$ overloaded from $h_k$, the number of non-$\failSuperset$ balls of Types 3 and 2 that are in bin $i$ after the second phase (or at any point during the swapping process) can be at most
$$2 (p_3 + \epsilon)\mu + 2 (p_2 + \epsilon)\mu  < \mu.$$
Therefore, the only way for bin $i$ to have more than $\mu$ non-$\failSuperset$ balls at the end of the second phase is if its Type $1$ balls (all of which are using hash $h_1$) are not in $\failSuperset$. We will therefore take as given for the rest of the proof that the Type $1$ balls in bin $i$ are not in $\failSuperset$, meaning that $z_i$ is just the number of Type 1 ball in bin $i$. 

Let $w_i$ be the number of non-$\failSuperset$ Type 2 balls $x$ satisfying $h_2(x) = i$. Since $\failSuperset$ contains all Type 2 balls $x$ for which bin $h_2(x)$ is Type 2 overloaded from $h_2$, we know that
$$w_i \le (p_2 + \epsilon) \mu.$$
Since $\failSuperset$ contains all Type 2 balls that hash to any Type 1 underloaded bins, the only way for $w_i$ to be non-zero is if 
$$z_i \ge (p_1 - \epsilon) \mu.$$
Therefore, $z_i \ge w_i$. This implies that at least $w_i$ Type 1 balls are removed from bin $i$ during the second phase. Since the number of non-$\failSuperset$ balls at the beginning of the phase was at most $\mu$, since $w_i$ non-$\failSuperset$ balls were added to the bin, and since at least $w_i$ non-$\failSuperset$ balls were removed from the bin, we can conclude that the number of non-$\failSuperset$ balls in the bin remains at most $\mu$ at the end of the phase. 
\end{proof}

Next, we state two lemmas to analyze the various graphs to which we apply the Canonical Orientation Procedure -- roughly speaking, these lemmas tell us that the graphs at hand are ``valid candidates'' for the procedure. (The same lemmas will also be useful later on when we are bounding recourse.) Both lemmas require nontrivial machinery to prove, and we defer this machinery to Section \ref{sec:graphs}.

For the next lemma, recall that for a graph $G(B)$, we define the random subgraphs $G_1(B),\ldots,G_{p^{-1}}(B)$ which partition the edges in $G(B)$ (each edge in $G(B)$ is assigned to one subgraph at random).

\begin{lemma}
For each $j \in [p^{-1}]$, we have with high probability in $n$ that the graph $G_j(\bafter(S))$ satisfies the following properties:
\begin{itemize}
    \item each connected component contains at most $O(1)$ cycles;
    \item the expected size of the connected component containing a given vertex $v \in [n]$ is $O(1)$.
\end{itemize}
\label{lem:goodgraphsa}
\end{lemma}
\begin{proof}
Since $|\bafter(S)| = |B(S)|$ is $O(n)$ with high probability, we have by Lemma \ref{lem:randomG} that, with high probability, $G$ is a graph with $O(n)$ independent and uniformly random edges from $[n] \times [n]$. The result therefore follows from Lemma \ref{lem:G1}. 
\end{proof}

\begin{lemma}
For each $j \in [p^{-1}]$, for each graph $G_j(X^{(a, b, c, d)})$ (where $a,b \in [3]$, $a > b$, and $c, d \in [2]$), and for each graph $G_j(Y^{(a, c)})$ (where $a \in [3]$ and $c \in [2]$), we have with high probability in $n$ that the graph satisfies the following properties:
\begin{itemize}
    \item each connected component contains at most $O(1)$ cycles;
    \item the expected size of the connected component containing a given vertex $v \in [n]$ is $O(1)$; 
    \item for any ball $s \in S$, the expected size of the connected component containing vertex $h_1(s)$ is also $O(1)$. 
\end{itemize}
\label{lem:goodgraphsb}
\end{lemma}
\begin{proof}
Let $m_1, m_2, m_3$ be the number of Type 1, Type 2, and Type 3 balls in the system, respectively. First note that, by Chernoff bounds, we have with high probability in $n$ that $m_1 = (1 \pm o(1)) p_1 m$, that $m_2 = (1 \pm o(1))p_2m$, and that $m_3 = (1 \pm o(1)) p_3 m$. 

It follows that each of the graphs in the lemma fits the assumptions of one of Lemmas \ref{lem:G2} (for the $X^{(a, b, c, d)}$s) or \ref{lem:G3} (for the $Y^{(a, c)}$s). The claimed result therefore follows from those lemmas (for the first two claimed bullet points) and Corollary \ref{cor:G} (for the third claimed bullet point).

\end{proof}

Putting the previous lemmas together, we can now bound the overload of our algorithm:
\begin{proposition}
    With high probability in $n$, the final state produced by our algorithm has overload $O(1)$.
    \label{prop:finaloverload}
\end{proposition}
\begin{proof}
By Lemmas \ref{lem:goodgraphsa} and \ref{lem:goodgraphsb}, with high probability, for each $j \in [p^{-1}]$, the graphs $G_j(B')$, the graphs $G_j(X^{(a, b, c,d)})$, and the graphs $G_j(Y^{(a, c)})$ all have the property that each connected component contains $O(1)$ cycles. It follows that our applications of the ECO Procedure result in an orientation in which each bin has at most $O(1)$ balls from $\failSuperset \cup B'$. By Lemma \ref{lem:Fsafe}, each bin contains at most $m/n$ balls that are neither in $\failSuperset$ nor in $B'$. Thus, we have with high probability that each bin contains a total of $m/n + O(1)$ balls. 
\end{proof}

\subsection{Bounding Recourse}

To complete the proof of Theorem \ref{thm:postprocess}, we must also prove a bound on recourse. We begin by analyzing the recourse within a given graph $G_j(X^{(a, b, c, d)})$.
\begin{lemma}
Let $j \in [p^{-1}]$ and consider one of the sets $X^{(a, b, c, d)}(S)$. Let $S$ and $S'$ be neighboring sets of $m$ balls, with $S = S' \setminus \{x_1\} \cup \{x_2\}$. Then, 
$$\E[|\overline{G}_j(X^{(a, b, c, d)}(S)) \triangle \overline{G}_j(X^{(a, b, c, d)}(S'))|] \le O(1).$$
That is, the expected number of edges on which orientations $\overline{G}_j(X^{(a, b, c, d)}(S))$ and $\overline{G}_j(X^{(a, b, c, d)}(S'))$ differ is $O(1)$. 
\label{lem:swappingrecourse0}
\end{lemma}
\begin{proof}
The orientations $\overline{G}_j(X^{(a, b, c, d)}(S))$ and $\overline{G}_j(X^{(a, b, c, d)}(S'))$ differ at most on the edges in the following two connected components:
\begin{itemize}
    \item the component in $\overline{G}_j(X^{(a, b, c, d)}(S))$ containing $x_1$;
    \item and the component in $\overline{G}_j(X^{(a, b, c, d)}(S'))$ containing $x_2$.
\end{itemize}
By Lemma \ref{lem:goodgraphsb}, these components each have expected size $O(1)$.
\end{proof}

By the same argument, we can analyze the recourse within a given graph $G_j(Y^{(a, c)})$. 
\begin{lemma}
Let $j \in [p^{-1}]$ and consider one of the sets $Y^{(a, c)}(S)$. Let $S$ and $S'$ be neighboring sets of $m$ balls, with $S = S' \setminus \{x_1\} \cup \{x_2\}$. Then, 
$$\E[|\overline{G}_j(Y^{(a, c)}(S)) \triangle \overline{G}_j(Y^{(a, c)}(S'))|] \le O(1).$$
That is, the expected number of edges on which orientations $\overline{G}_j(Y^{(a, c)}(S))$ and $\overline{G}_j(Y^{(a, c)}(S'))$ differ is $O(1)$. 
\label{lem:swappingrecourse1}
\end{lemma}
\begin{proof}
The orientations $\overline{G}_j(Y^{(a, c)}(S))$ and $\overline{G}_j(Y^{(a, c)}(S'))$ differ at most on the edges in the following two connected components:
\begin{itemize}
    \item the component in $\overline{G}_j(Y^{(a, c)}(S))$ containing $x_1$;
    \item and the component in $\overline{G}_j(Y^{(a, c)}(S'))$ containing $x_2$.
\end{itemize}
By Lemma \ref{lem:goodgraphsb}, these components each have expected size $O(1)$.
\end{proof}
Finally, by a slightly more intricate argument, we can bound the expected recourse from changes to the set $\bint(S)$ and from changes to the orientation $\overline{G}(B'(S))$. 
\begin{lemma}
Consider two neighboring sets $S$ and $S'$ of $m$ balls. Then, $$\E[|\bint(S) \triangle \bint(S')|] \le O(\recourse[\alg]),$$
and the expected number of edges on which the orientations $\overline{G}_1(\bafter(S))$ and $\overline{G}_2(\bafter(S'))$ disagree is \\$O(\recourse[\alg])$. 
\label{lem:swappingrecourse2}
\end{lemma}
\begin{proof}
Let us imagine that $\alg(S)$ is determined by an adversary as follows: 

\begin{enumerate}
    \item The adversary selects sets $S_1, S_2$ of Type 1 and Type 2 balls, and chooses $h_1(s)$ for each ball $s \in S_1 \cup S_2$. Each ball $s \in S_1 \cup S_2$ is then placed into bin $h_1(s)$. Next, the adversary selects numbers $x_i$, $i \in [n]$, where each $x_i$ represents the number of Type 3 balls that go in bin $i$. Call $S_1, S_2, h_1(S_1), h_1(S_2), \{x_i\}$ the \defi{core input} $C$. 
    \item Once the core input is determined, the hashes $h_2(S_1)$ and $h_2(S_2)$ are determined, and the random tapes used in the Two-Phase Swapping Procedure are generated. We refer to the randomness determined in this step as the \defi{Shuffle Randomness}.
    \item Finally, the Two-Phase Swapping Procedure is performed, generating $\bint, \bafter$. We will now think of these as being a function of the core input $C$, rather than a set $S$ of balls.
\end{enumerate}

By Lemma \ref{lem:randomG}, if we fix any core input $C$, the remaining randomness (i.e., the shuffle randomness) is enough to ensure that the hashes of the balls in $\bafter(C)$ are independent and uniformly random. 

To complete the proof, we will consider what happens if two core inputs $C$ and $C'$ differ from each other by only a small amount. Say that two core inputs $C$ and $C'$ differ by a \defi{unit change} if we can get from $C$ to $C'$ (or possibly from $C'$ to $C$) by one of the following three basic changes:
\begin{itemize}
    \item Incrementing/decrementing $x_i$ for some bin $i$.
    \item Adding/removing some ball $b$ to $S_2$ with $h_1(b) = i$ for some $i$.
    \item Adding/removing some ball $b$ to $S_1$ with $h_1(b) = i$ for some $i$.
\end{itemize}
To complete the proof, it suffices to show that, when $C$ and $C'$ differ by a unit change, 
$$\E[|\bafter(C) \triangle \bafter(C')|] \le O(1)$$
and the expected number of edges on which the orientations $\overline{G}_1(\bafter(C))$ and $\overline{G}_1(\bafter(C'))$ disagree is $O(1)$. 

We now prove this for each type of unit change, where the change takes us from $C$ to $C'$:
\begin{itemize}
    \item \textbf{Case 1: Some $x_i$ is incremented/decremented.} Without loss of generality, $x_i$ is incremented. If bin $i$ has $\le h$ balls in $C'$, then $\bint(C) = \bint(C')$ and $\bafter(C) = \bafter(C')$. Otherwise, $\bint(C') = \bint(C) \cup \{x_1\}$ for some ball $x_1$, where $h_1(x_1) = i$ and where $h_2(x_1)$ is determined by the Shuffle Randomness; and $\bafter(C') = \bafter(C) \cup \{x_2\}$ for some ball $x_2$ where $h_1(x_2) = h_2(x_1)$ and where $h_2(x_2)$ is determined by the Shuffle Randomness. Note that the values $j \coloneqq h_1(x_2) = h_2(x_1)$ and  $k \coloneqq h_2(x_2)$ are independent of the randomness used to construct graph $G_1(\bafter(C))$.

    By Lemma \ref{lem:goodgraphsb}, the expected size of the connected components in $G_1(\bafter(C))$ containing nodes $j$ and $k$ are each $O(1)$. It follows that the connected component in $G_1(\bafter(C'))$ containing edge $(j, k)$ is $O(1)$. As this component contains all edges on which the orientations $\overline{G}_1(\bafter(C))$ and $\overline{G}_1(\bafter(C'))$ disagree, the expected number of such edges is $O(1)$. 
    
    \item \textbf{Case 2: Some ball $b$ is added/removed to $S_2$ with hash $h_1(i)$ for some $i$: } Without loss of generality the ball $b$ is added to $C$ (rather than removed). If bin $h_1(b)$ has $\le h$ balls in $C'$, then $\bint(C) = \bint(C')$ and $\bafter(C) = \bafter(C')$. Otherwise, $\bint(C') = \bint(C) \cup \{x_1\}$ for some ball $x_1$ (possibly $b$), where $h_1(x_1) = h_1(b)$ and where $h_2(x_1)$ is determined by the Shuffle Randomness; and $\bafter(C') = \bafter(C) \cup \{x_2\}$ for some ball $x_2$ where $h_1(x_2) = h_2(x_1)$ and where $h_2(x_2)$ is determined by the Shuffle Randomness. Note that the values $j \coloneqq h_1(x_2) = h_2(x_1)$ and  $k \coloneqq h_2(x_2)$ are independent of the randomness used to construct graph $G_1(\bafter(C))$.

    By Lemma \ref{lem:goodgraphsb}, the expected size of the connected components in $G_1(\bafter(C))$ containing nodes $j$ and $k$ are each $O(1)$. It follows that the connected component in $G_1(\bafter(C'))$ containing edge $(j, k)$ is $O(1)$. As this component contains all edges on which the orientations $\overline{G}_1(\bafter(C))$ and $\overline{G}_1(\bafter(C'))$ disagree, the expected number of such edges is $O(1)$.

    \item \textbf{Case 3: Some ball $b$ is added to $S_1$ with $h_1(b) \coloneqq i$ for some $i$:} To simplify the analysis of this option, let us decompose this modifiction into two sub-modifications. Submodification 1 increments $x_i$. Then, Submodification 2 inserts the ball $b$ and decrements $x_i$. Since Submodification 1 has already been analyzed in Case 1, we can focus here on analyzing Submodification 2. That is, we redefine $C$ to be the core input before Submodification 2, and $C'$ to be the core input after. 
    
    Because Submodification 2 does not change the number of balls in any given bin, we have $\bint(C) = \bint(C')$. Moreover, either $\bafter(C) = \bafter(C')$, or $\bafter(C') = \bafter(C) \setminus \{b'\} \cup \{b\}$ for some ball $b' \in S_1$ such that $h_1(b) = i$. 

    By Lemma \ref{lem:goodgraphsb}, the expected size of the connected component in $G_1(\bafter(C))$ containing vertex $i$ is $O(1)$, and the expected size of the connected component $G_1(\bafter(C'))$ containing vertex $i$ is also $O(1)$. Since every edge on which the orientations $\overline{G}_1(\bafter(C))$ and $\overline{G}_1(\bafter(C'))$ disagree is in at least one of those two components, the expected number of disagreements between the two orientations is $O(1)$. 
\end{itemize}
\end{proof}
Putting the pieces together, we can establish our desired bound on recourse:
\begin{proposition}
The final algorithm has expected recourse $O(\recourse[\alg])$.
\label{prop:finalrecourse}
\end{proposition}
\begin{proof}
Consider two neighboring sets $S$ and $S'$. The expected recourse of our algorithm is at most the expected recourse of $\recourse[\alg]$ from running $\alg$, plus the following terms:
\begin{enumerate}
    \item $\E[|\bint(S) \triangle \bint(S')|]$;
    \item $\E[|\overline{G}_i(\bafter(S)) \triangle \overline{G}_i(\bafter(S'))|]$ for each $i \in [p^{-1}]$;
    \item $\E[|\overline{G}_i(X^{(a, b, c, d)}(S)) \triangle \overline{G}_i(X^{(a, b, c, d)}(S'))|]$ for each $i \in [p^{-1}]$ and each $a, b \in [3]$ with $a > b$ and $c, d \in [2]$;
    \item $\E[|\overline{G}_i(Y^{(a, c)}(S)) \triangle \overline{G}_i(Y^{(a, c)}(S'))|]$ for each $i \in [p^{-1}]$ and $a \in [3]$ and $c \in [2]$. 
\end{enumerate}
By Lemmas \ref{lem:swappingrecourse0}, \ref{lem:swappingrecourse1}, and \ref{lem:swappingrecourse2}, along with the fact that $p^{-1} = O(1)$, we have that these terms sum to $O(\recourse[\alg])$, as desired.
\end{proof}

Combining our analyses of overload (Proposition \ref{prop:finaloverload}) and recourse (Proposition \ref{prop:finalrecourse}), the proof of Theorem \ref{thm:postprocess} is complete.

\subsection{Machinery for Analyzing Graphs}\label{sec:graphs}

Finally, in this section, we develop the machinery needed to analyze the graphs constructed in the previous section. Specifically, the results that we develop in this section are what allow us to prove Lemmas \ref{lem:goodgraphsa} and \ref{lem:goodgraphsb} in Section \ref{sec:postprocessoverload}.

We begin with the following lemma:

\begin{lemma}
Let $c$ be a sufficiently small positive constant, and let $0 \le k \le O(1)$. Consider a directed graph $G$ on $n$ vertices, where each vertex $v \in [n]$ has $X_v$ outgoing edges, where each $X_v$ independently satisfies $\Pr[X_v \ge j] \le c^j$ for all $j \ge 0$, and where each edge $(v, u)$ has an independent and uniformly random endpoint $u \in [n]$. Consider the (undirected) connected components of $G$. For any given vertex $v$ and for $j \ge 1$ and $k \ge 0$, the probability that $v$ is part of a connected component $C$ of size $|C| = j$ containing at least $j - 1 + k$ edges is at most 
$$e^{-\Omega(j)} / n^{k / 2} + 1/n^{\omega(1)}.$$
\end{lemma}
\begin{proof}
Since $C$ must contain $v$, there are $\binom{n - 1}{j - 1} \le n^{j - 1} e^{j - 1} / j^{j - 1}$ options for $C$. For a given choice of $C$, let $Y$ be the random variable denoting the sum of the out-degrees of the vertices in $C$. In order for $C$ to have $j - 1 + k$ edges, we would need $Y = j - 1 + k$. And, in order for $C$ to be a connected component, we would need all those edges to have right endpoints in $C$, an event that happens with probability at most
$$(|C| / n)^{j - 1 + k} = (j / n)^{j - 1 + k}.$$
Thus, the probability of $v$ is part of a connected component $C$ of size $|C| = j$ containing at least $j - 1 + k$ edges is at most 
$$n^{j - 1} e^{j - 1} / j^{j - 1} \cdot \Pr[Y = j - k + 1] \cdot (j / n)^{j - 1 + k}.$$
By a Chernoff bound for sums of independent geometric random variables\footnote{Here, we are using the following bound. Let $X = X_1 + \cdots + X_j$ be a sum of iid geometric random variables, where each $X_i$'s expectation is at most a sufficiently small positive constant. For any $m \ge j/2$, we have $\Pr[X \ge m] \le 1/8^m$. }, we have $\Pr[Y = j - 1 + k] \le 1/8^{j - 1 + k}$. The entire expression is therefore at most 
\begin{align*}
& n^{j - 1} e^{j - 1} / j^{j - 1} \cdot 8^{-(j - 1 + k)} \cdot (j / n)^{j - 1 + k} \\
& = O(e^{j - 1} 8^{-j + k} (j / n)^{k}). \\
& e^{-\Omega(j)} / n^{k / 2} + 1/ n^{\omega(1)}.
\end{align*}

\end{proof}

As an immediate corollary, we get the following:
\begin{corollary}
    Let $c$ be a sufficiently small positive constant, and let $0 \le k \le O(1)$. Consider a directed graph $G$ on $n$ vertices, where each vertex $v \in [n]$ has $X_v$ outgoing edges, where each $X_v$ independently satisfies $\Pr[X_v \ge j] \le c^j$ for all $j \ge 0$, and where each edge $(v, u)$ has an independent and uniformly random endpoint $u \in [n]$. Finally, let $G'$ be the undirected version of the graph.
    
    With high probability in $n$, every connected component in $G'$ contains at most $O(1)$ cycles. Moreover, for a given vertex $i$, the expected size of the component containing $i$ is $O(1)$ edges. 
    \label{cor:cuckooproperties}
\end{corollary}

We will also need the following helper claim, which analyzes the effect of down-sampling the edges at a given vertex:

\begin{lemma}
Let $0 < p < 1$ be a parameter. Suppose we have $X$ balls, where $\Pr[X \ge i] \le e^{-\Omega(i)}$ for all positive $i$. Sample each of these balls independently with probability $p$, and let $Y$ be the number of sampled balls. Then, for all positive $i$,
$$\Pr[Y \ge i] \le O(1) \cdot p^{i/2}.$$
\label{lem:downsample}
\end{lemma}
\begin{proof}
We can upper bound $\Pr[Y \ge i]$ by
\begin{align*}
    & \sum_{j \ge i} \Pr[X = j] \sum_{S \subseteq [j], |S| = i} \Pr[S \text{ sampled}] \\
    & = \sum_{j \ge i} \Pr[X = j] \binom{j}{i} p^i \\
    & \le \sum_{j \ge i} e^{-\Omega(j)} \cdot (e^i (j / i)^i) \cdot p^i. \\
\end{align*}

We can complete the proof by demonstrating that each summand is at most $e^{-\Omega(j)} \cdot c^{i/2}$, as this would imply that the entire sum is at most $O(1) \cdot p^{i/2}$.

We can bound each summand by considering two cases. For $j = \Theta(i)$, the $c^i$ term dominates the $e^i (j / i)^i$ term so that $e^{-\Omega(j)} \cdot (e^i (j / i)^i) \cdot c^i \le e^{-\Omega(j)} c^{i / 2}$. For $j = \omega(i)$, the $e^{-\Omega(j)}$ term dominates the $e^i (j / i)^i$ term so that $e^{-\Omega(j)} \cdot (e^i (j / i)^i) \cdot c^i \le e^{-\Omega(j)} c^{i}$. Either way, the entire summand is at most $e^{-\Omega(j)}\cdot c^{i/2}$, as desired.
\end{proof}

Finally, we will want to apply Corollary \ref{cor:cuckooproperties} in settings where the vertices of the graph do not quite behave independently. To rectify this issue, we will make use of the following Poissonization lemmas: 
\begin{lemma}[Poisson Approximation for Expectations \cite{mitzenmacher2017probability}]
Let $f(a_1, \ldots, a_n): (\mathbb{N} \cup \{0\})^n \rightarrow \mathbb{R}$ be a monotonic function (either monotonically increasing or monotonically decreasing). Throw $n$ balls (resp.~$\operatorname*{Poisson}(n)$ balls) independently and uniformly into $n$ bins, and let $a_i$ (resp.~$a_i'$) denote the number of balls in the $i$-th bin. Then, 
$$\E[f(a_1, \ldots, a_n)] \le 2\E[f(a_1', \ldots, a_n')].$$
\label{lem:poissonexp}
\end{lemma}

As a slight abuse of notation, although technically the function $f$ in Lemma \ref{lem:poissonexp} is monotonic in the quantities $a_1, a_2, \ldots, a_n$, we will say as a shorthand that it is monotonic with ball throws. 

\begin{lemma}[Poisson Approximation for Probabilities \cite{mitzenmacher2017probability}]
Let $E(a_1, \ldots, a_n)$ be an event determined by non-negative integers $a_1, \ldots, a_n$. Throw $n$ balls (resp.~$\operatorname*{Poisson}(n)$ balls) independently and uniformly into $n$ bins, and let $a_i$ (resp.~$a_i'$) denote the number of balls in the $i$-th bin. Then, 
$$\Pr[E(a_1, \ldots, a_n)] \le O(\sqrt{n}) \cdot \Pr[E(a_1', \ldots, a_n')].$$
It follows that, if the event $E(a_1', \ldots, a_n')$ occurs with high probability in $n$, then the event $E(a_1, \ldots, a_n)$ also occurs with high probability in $n$. 
\label{lem:poissonprob}
\end{lemma}

Finally, with the preceding lemmas in place, we are prepared prove the main results of the section. Each of the following lemmas analyzes one of the random processes that we use to construct random graphs in our algorithm -- in each case, we argue that the resulting random graph is well behaved.

\begin{lemma}
Let $d \in O(1)$, and let $p$ be a sufficiently small positive constant relative to $d$. Suppose we throw $cn$ balls into $n$ bins, independently and uniformly at random, and then we remove each ball from the system independently with probability $1 - p$. Let $Z_i$ be the number of remaining balls in bin $i$. 

Construct a (multi-)graph $G$ on vertices $[n]$ by adding for each vertex $i \in [n]$ $Z_i$ edges from $i$ to independent and uniformly random vertices in $[n]$. Then, $G$ satisfies the following two properties:
\begin{itemize}
    \item With high probability in $n$, every connected component in $G$ contains $O(1)$ cycles.
    \item For any given vertex $i$, the expected size of the connected component $C_i$ containing $i$ is $O(1)$.
\end{itemize}
\label{lem:G1}
\end{lemma}
\begin{proof}
By Poissonization (Lemmas \ref{lem:poissonexp} and \ref{lem:poissonprob}), we may assume without loss of generality that, rather than $m$ balls being thrown, $M \sim \operatorname*{Poisson}(m)$ balls are thrown. (Here we use that $|C_i|$ monotonically increases with ball throws.) With this modification, the bins are now independent, each independently receiving $\operatorname*{Poisson}(d)$ balls. Since the $Z_i$s are obtained by sampling each ball with probability $P$, it follows that the $Z_i$s are independent $\operatorname*{Poisson}(pd)$ random variables. Supposing that $p$ is a sufficiently small positive constant, this implies that $Z_i$s independently satisfy 
\begin{equation}\Pr[Z_i \ge j] \le c^j,
\label{eq:Zsmall2}
\end{equation}
where $c$ is a positive constant determined by $p$ that goes to $0$ as $p$ goes to $0$. Applying Corollary \ref{cor:cuckooproperties}, we obtain the desired properties of the graph $G$. 
\end{proof}

\begin{lemma}
Let $h \ge 1$ and let $\epsilon \in (0, 1)$ be a positive constant, and let $p$ be a sufficiently small positive constant as a function of $\epsilon$. Suppose we throw $m = h n$ balls into $n$ bins, independently and uniformly at random. Let $X_i$ denote the number of balls in bin $i$. Now suppose that we remove each ball from the system independently with probability $1 - p$. Let $Y_i$ be the number of remaining balls in bin $i$, and define 
$$Z_i = \begin{cases}
0 & \text{ if } X_i \le (1 + \epsilon) h \\
Y_i & \text{otherwise.}
\end{cases}
$$
Finally, construct a (multi-)graph $G$ on vertices $[n]$ by adding for each vertex $i \in [n]$ $Z_i$ edges from $i$ to independent and uniformly random vertices in $[n]$. Then, $G$ satisfies the following two properties:
\begin{itemize}
    \item With high probability in $n$, every connected component in $G$ contains $O(1)$ cycles.
    \item For any given vertex $i$, the expected size of the connected component $C_i$ containing $i$ is $O(1)$.
\end{itemize}
\label{lem:G2}
\end{lemma}
\begin{proof}
By Poissonization (Lemmas \ref{lem:poissonexp} and \ref{lem:poissonprob}), we may assume without loss of generality that, rather than $m$ balls being thrown, $M \sim \operatorname*{Poisson}(m)$ balls are thrown. (Here we use that $|C_i|$ monotonically increases with ball throws.) With this modification, the bins are now independent, each independently receiving $\operatorname*{Poisson}(m/n)$ balls. 

Define $\overline{X}_i$ to be $0$ if $X_i \le (1 + \epsilon)h = (1 + \epsilon)m/n$ and to be $X_i$ otherwise. For any $j \ge 0$, we have
\begin{align*}
\Pr[X_i \ge j] &= \begin{cases} 0 & \text{ if } j \le (1 + \epsilon)h \\ \Pr[\operatorname*{Poisson}(m/n) \ge j] & \text{ otherwise} \end{cases} \\
& \le e^{-\Theta(j)},
\end{align*}
since $\Pr[\operatorname*{Poisson}(\lambda) \ge j] \le e^{-\Omega(j)}$ for all $j \ge (1 + \Omega(1)) \lambda$.

Since $Z_i$ is obtained from $\overline{X}_i$ by sampling every ball with probability $p$, it follows from Lemma \ref{lem:downsample} that
\begin{equation}\Pr[Z_i \ge j] \le c^j,
\label{eq:Zsmall1}
\end{equation}
where $c$ is a positive constant determined by $p$ that goes to $0$ as $p$ goes to $0$. 

Combining \eqref{eq:Zsmall1} with the fact that the $Z_i$s are independent (thanks to Poissonization), we can apply Corollary \ref{cor:cuckooproperties} to obtain the desired properties of the graph $G$. 
\end{proof}

\begin{lemma}
Let $h_1 \ge h_2$, let $\epsilon \in (0, 1)$ be a positive constant, and let $p$ be a sufficiently small positive constant as a function of $\epsilon$. Suppose we throw $m_1 = h_1 n$ blue balls into $n$ bins and $m_2 = h_2 n$ red balls in $n$ bins, independently and uniformly at random. Let $X_i$ denote the number of blue balls in bin $i$. Now suppose that we remove each blue ball from the system independently with probability $1 - p$. Let $Y_i$ be the number of remaining red balls in bin $i$, and define 
$$Z_i = \begin{cases}
0 & \text{ if } X_i \ge (1 - \epsilon) h_1 \\
Y_i & \text{otherwise.}
\end{cases}
$$
Finally, construct a (multi-)graph $G$ on vertices $[n]$ by adding for each vertex $i \in [n]$ $Z_i$ edges from $i$ to independent and uniformly random vertices in $[n]$. Then, $G$ satisfies the following two properties:
\begin{itemize}
    \item With high probability in $n$, every connected component in $G$ contains $O(1)$ cycles.
    \item For any given vertex $i$, the expected size $|C_i|$ of the connected component containing $i$ is $O(1)$.
\end{itemize}
\label{lem:G3}
\end{lemma}
\begin{proof}
By applying Poissonization (Lemmas \ref{lem:poissonexp} and \ref{lem:poissonprob}) to the blue balls, we may assume without loss of generality that, rather than $m_1$ blue balls being thrown, $M_1 \sim \operatorname*{Poisson}(m_1)$ blue balls are thrown. (Here we use that $|C_i|$ monotonically decreases with blue ball throws.) By then applying Poissonization again, but this time to the red balls, we may also assume that the number of red balls is not $m_2 = h_2n$ but is instead $M_2 \sim \operatorname*{Poisson}(m_2)$. (Here we use that $|C_i|$ monotonically increases with red ball throws.)

Define $W_i$ to be $0$ if $X_i \ge (1 - \epsilon)h_1$ and be the number of red balls in bin $i$ otherwise. For any $j \ge 0$, we have
\begin{align*}
\Pr[W_i \ge j] & = \Pr[\operatorname*{Poisson}(h_1) \le (1 - \epsilon) h_1] \cdot \Pr[\operatorname*{Poisson}(h_2) \ge j] \\ & \le \min(\Pr[\operatorname*{Poisson}(h_1) \le (1 - \epsilon)h_1], \Pr[\operatorname*{Poisson}(h_2) \le j]).
\end{align*}
For $j \ge 2h_1 \ge 2h_2$, we have $\Pr[\operatorname*{Poisson}(h_2) \ge j] \le e^{-\Omega(j)}$, and for $j \le 2h_1$, we have $\Pr[\operatorname*{Poisson}(h_1) \le (1 - \epsilon) h_1] \le e^{-\Omega(h_1)} = e^{-\Omega(j)}$. Thus, for all values of $j \ge 0$, we have
\begin{align*}
\Pr[W_i \ge j] \le e^{-\Omega(j)}. 
\end{align*}

Since $Z_i$ is obtained from $W_i$ by sampling every ball with probability $p$, it follows from Lemma \ref{lem:downsample} that
\begin{equation}\Pr[Z_i \ge j] \le c^j,
\label{eq:Zsmall}
\end{equation}
where $c$ is a positive constant determined by $p$ that goes to $0$ as $p$ goes to $0$. 

Combining \eqref{eq:Zsmall} with the fact that the $Z_i$s are independent (thanks to Poissonization), we can apply Corollary \ref{cor:cuckooproperties} to obtain the desired properties of the graph $G$. 
\end{proof}

Finally, we also establish an extension of the previous three lemmas, which allows the vertex $i$ that we are considering to be selected not arbitrarily but by a specific ball. 
\begin{corollary}
 Consider the setups in any of Lemmas \ref{lem:G1}, \ref{lem:G2}, and \ref{lem:G3}. Condition on some specific ball $b$ being thrown into bin $1$. Then, the expected size of the connected component containing vertex $1$ remains $O(1)$. 
 \label{cor:G}
\end{corollary}
\begin{proof}
We can follow the same proofs as before, but excluding $b$ from the balls that we Poissonize (i.e., applying Poissonization to all of the balls \emph{except} for $b$). Additionally, when establishing \eqref{eq:Zsmall2}, \eqref{eq:Zsmall1}, and \eqref{eq:Zsmall}, a small amount of additional care is needed to note that the deterministic presense of ball $b$ in bin $1$ does not violate the correctness of any of the three inequalities. 
\end{proof}

\subsection{Proving Theorem~\ref{thm:full-theorem}}
We are now ready to prove Theorem~\ref{thm:full-theorem}.
\thmfulltheorem*
\begin{proof}
    By Theorem~\ref{thm:offline-slice-and-spread}, Slice and Spread has cumulative overload $O(n)$ with high probability in $n$ and expected recourse $O(\log\log (m/n))$. Furthermore, Slice and Spread is a good pre-baking algorithm by construction. Applying Theorem~\ref{thm:postprocess} completes the proof.
\end{proof}
\section{Conjectures and Open Questions}

In this paper, we present a history-independent two-choice algorithm that achieves $O(1)$ expected overload and $O(\log \log (m/n))$ expected recourse. We conjecture that these bounds are optimal, and, in particular, that $O(1)$ overload requires $\Omega(\log \log (m/n))$ recourse. On the lower-bound side, even proving any $\omega(1)$ bound on expected recourse (for either weakly or strongly history independent solutions) would be quite interesting. 

Even if one does not care about history independence, the results in this paper offer a new state of the art for dynamic two-choice load balancing with recourse. Here, again, it is an interesting open question whether one can achieve $O(1)$ expected overload with $o(\log \log (m/n))$ expected recourse. Note that, in general, one of the main challenges for this problem is the following: One must be careful to handle \emph{reappearance dependencies} \cite{bansal2022balanced,agrawal2024distributed}, which occur if the same element is inserted, deleted, and resinserted many times, in which case the random hashes used by the element have already affected the state of the data structure in the past, and therefore can no longer be treated as random. History independence dodges this issue by eliminating history-related dependencies entirely. It is possible that history-dependent algorithms may be able to do better, but such algorithms would likely require nontrivial techniques to handle reappearance dependencies. 

Finally, closely related to the problem of constructing dynamic two-choice load-balancing strategies is the problem of designing time-efficient implementations of bucketized cuckoo hashing \cite{dietzfelbinger2007balanced,frieze2018balanced}. Here, one cares both about time efficiency and space efficiency (including the space needed to store metadata used by the algorithm). One of the most natural open questions in this area \cite{frieze2018balanced} is whether there exists a solution that supports space efficiency $1 - \epsilon$, using buckets of size $O(\epsilon^{-1})$, and with $O(1)$ expected-time insertions/deletions. This question, which is at least as hard as the dynamic two-choice load-balancing problem, remains open.

\section*{Acknowledgments}
This work was supported in part by NSF grants CCF-2247577, CCF-2106827, CNS-2504471, 2212746, 2044679, 2128519, the Packard fellowship, the Carnegie Mellon University CyLab Presidential Fellowship, a grant from 0xPARC, and a grant from the late Nikolai Mushegian.

\bibliographystyle{abbrv}
\bibliography{main, balls}

\appendix

\section{A Lower Bound for History Independent Greedy}\label{sec:greedylower}

In this section, we prove the following lower bound:
\propgreedylower*

Throughout the section, we will focus on a single insertion of an element $x^*$ into a set $\mathcal{S}$ of size $m - 1$, and we will choose $x^*$ to be smaller than all of the elements in $\mathcal{S}$ according to the canonical ordering used by HI Greedy. As in the upper bound, we consider two parallel worlds, called world 0 and world 1, in which we insert the balls $\mathcal{S}$ and $\mathcal{S} \cup \{x^*\}$ according to the canonical ordering. We will think of the insertions as being performed at the same time in the two worlds---time $i$ refers to the time immediately after the $(i - 1)$th insertion in world 0 and immediately after the $i$-th insertion in world $1$. 

By Fact \ref{fact:special}, we know that the only difference between the loads of the bins in these two worlds is that, at any given moment (after the first insertion), there is exactly one \defi{special bin} that contains one more ball in world 1 than it does in world 0. We will denote the load of the special bin at time $t$, including the extra ball that it contains in world 1, as $\special(t)$. 

Since the bin loads are essentially the same in the two world (with the exception of the special bin), we can feel free to focus our discussion on world 1. Throughout the rest of the section, when we refer to the load of a bin, we are implicitly referring to its load in world 1. 

As notation, we will use $\load{t}{i}$ to refer to the load of bin $i$ at time $i$. We will also let $x_1 = x^*$, and let $x_2, x_3, \ldots, x_{m - 1}$ denote the balls in $\mathcal{S}$ in the canonical ordering used to insert them.

Call a pair $(i, j) \in [n] \times [n]$ a \defi{critical tie} at time $t$ if, at time $t$, bin $i$ is the special bin and $\load{t}{j} = \special(t) - 1$. 
\begin{lemma}
If, immediately prior to ball $x_t$ being inserted, one of the pairs $(h_1(x_t), h_2(x_t))$ or $(h_2(x_t), h_1(x_t))$ is a critical tie, then with 50\% probability the ball $x_t$ gets allocated to a different bin in world 1 than in world 0. 
\label{lem:tietorecourse}
\end{lemma}
\begin{proof}
Without loss of generality, bin $h_1(x_t)$ is the special bin and bin $h_2(x_t)$ has load one smaller than that of the speical bin at time $t - 1$. Then, in world $0$ the insertion sees equal loads in bins $h_1(x_t)$ and $h_2(x_t)$ and picks one at random; but in world $1$, the insertion sees a larger load in bin $h_1(x_t)$ and picks bin $h_2(x_t)$ deterministically.
\end{proof}

Say that a critical tie $(i, j)$ is \defi{created} at time $t$ if the critical tie occurs at time $t$ but not at time $t - 1$. Using Lemma \ref{lem:tietorecourse}, we can relate the total number of critical ties that are created (over all time) to the total recourse incurred by HI Greedy.

\begin{lemma}
Let $Q$ be the total number of critical ties created across all insertions, and let $R$ be the total recourse of the algorithm. Then,
$$\E[R] \ge \Omega(\E[Q] /n) - O(1).$$
\label{lem:tiestorecourse}
\end{lemma}
\begin{proof}
To simplify the proof, let us assume that the final ball $x_m$ to be inserted has hash $h_1(x_m)$ equal to whatever the special bin is at time $m - 1$. Note that this assumption is without loss of generality, as it changes $R$ and $Q$ each by at most one. Call this assumption the \emph{Nice Final Insertion Assumption}.

Recall that a critical tie is created when a bin $j$ has load $\special(t) - 1$ (and when this was not true for time $t - 1$). Suppose this happens at some time $t < m$, and let $k$ be the special bin at that time. Consider the next insertion $x$ to have at least one of its hashes $h_1(x), h_2(x)$ equal to at least one of $j$ or $k$ (such an $x$ must exist by the Nice Final Insertion Assumption). Given that at least one of the insertion's hashes are $j$ or $k$, the probability that the insertion's hashes are \emph{both} $j$ and $k$ is $\Omega(1/n)$. If this happens, then the insertion $x$ is incurs recourse with probability at least one half (by Lemma \ref{lem:tietorecourse}). 

The above process matches each critical tie that is created prior to time $m$ to a (distinct) future insertion that, with probability at least $\Omega(1/n)$, incurs recourse. It follows that $\E[R] \ge \E[Q] \cdot \Omega(1/n) - \Omega(1)$.
\end{proof}

Thus, to prove Proposition \ref{prop:greedylower}, it suffices to show that the expected number of critical ties that are created over all insertions is $\Omega(m)$. This means that a constant fraction of insertions (in expectation) create critical ties.

The intuition for why this would be the case comes from the following lemma, which says that, at any given moment, we expect most of the bins to have heights within a band of width $O(1)$:

\begin{lemma}
Suppose we perform $m$ insertions using the greedy algorithm. For any positive constant $c_1$, there exists a positive constant $c_2$ such that the following is true: after the $t$-th insertion, the expected number of bins with loads in $[t/n - c_2, t/n + c_2]$ is at least $(1 - 1/c_1)n$. 
\label{lem:bandexp}
\end{lemma}
\begin{proof}
This follows directly from Theorem 2.1 of \cite{talwar2014balanced}, which says that there exists a positive constant $a$ for which 
$$\E\left[\sum_{i = 1}^n e^{a \cdot |t/n - \load{t}{i}|}\right] \le O(n).$$
\end{proof}

With the help of Fact \ref{fact:special}, we can use McDiarmid's inequality (see Theorem \ref{thm:mcdiarmid} in Appendix \ref{app:mcdiarmid}) to get a high-probability version of Lemma \ref{lem:bandexp}.
\begin{lemma}
Suppose $m \le n^{2 - \Omega(1)}$, and suppose we perform $m$ insertions using the greedy algorithm. For any positive constant $c_1$, there exists a positive constant $c_2$ such that the following is true with high probability in $n$. After each insertion $t = 1, 2, \ldots m$, at least $1 - 1/c_1$ fraction of the bins have loads within $c_2$ of $t/n$. 
\label{lem:bandhigh}
\end{lemma}
\begin{proof}
It suffices to prove the bound for $t = m$, since it then follows by a union bound over $t \in [1, m]$. 

Define $X_1, X_2, \ldots, X_m$, where $X_i$ represents the hashes used by the $i$-th insertion. Let $F(X_1, X_2, \ldots, X_m)$ denote the number of bins with heights in the range $[t/n - c_2, t/n + c_2]$. By Fact \ref{fact:special}, $F$ is $2$-Lipschitz (i.e., changing a value of some $X_i$ changes $F$ by at most $2$). Since $X_1, X_2, \ldots, X_m$ are independent, it follows by McDiarmid’s inequality that 
$$\Pr[F - \E[F] \ge \sqrt{m} \log n] \le 1 / n^{\omega(1)}.$$
Since $m \le n^{2 - \Omega(1)}$, we have with high probability in $n$ that $F$ is within $o(m)$ of its mean.

To complete the proof, it suffices to show that $\E[F] \ge (1 - f(c_2)) n$ for some function $f$ that goes to $0$ as $c_2$ goes to infinity. This follows from Lemma \ref{lem:bandexp}.
\end{proof}

Building on this, we can also get the following lemma, which says that, although most bins have loads that are in a band of size $c_2$ around $t/n$, there will also still be a non-trivial $\Omega(1)$ fraction that are below the band. 
\begin{lemma}
Suppose $m \le n^{2 - \Omega(1)}$, and suppose we perform $m$ insertions using the greedy algorithm. For any sufficiently large positive constant $c_2$, there exists a positive constant $c_3$ such that the following is true with high probability in $n$. At any time $t \ge 3 c_2 n$, at least a $1/c_3$ fraction of the bins have loads smaller than $t/n - 2 c_2n$.
\label{lem:belowband}
\end{lemma}
\begin{proof}
Since $c_2$ is a sufficiently large positive constant, we know by Lemma \ref{lem:bandhigh} that with high probability in $n$ the following is true: at time $t - 3c_2n$, there are $\Omega(n)$ bins $A$ with heights less than $(t - 2 c_2) n$. During insertions $t - 3c_2n + 1, \ldots, t$, the probability that a given bin $a \in A$ manages to \emph{avoid} any new elements hashing to it is at least
$(1 - 1/n)^{O(n)} = \Omega(1)$. The expected number of bins that, at time $t$, have load less than $t - 2c_2 n$ is therefore at least $ \Omega(n)$. By McDiarmid’s Inequality (applied in the same way as in Lemma \ref{lem:bandhigh}), the number of such bins is within $o(n)$ of its mean with high probability in $n$. It follows that, with high probability in $n$, there are $\Omega(n)$ such bins. 
\end{proof}

The main difficulty in proving Proposition \ref{prop:greedylower} is that the $\special(t)$ is not simply  $\load{t}{i}$ for a uniformly random bin $i$. Rather, the load of the special bin evolves according to a more intricate process. The next lemma says that, nonetheless, we do expect $\special(t)$ to typically be within $O(1)$ of $n/t$. 

\begin{lemma}
Let $c$ be a sufficiently large positive constant. The expected number of times $t \in [m]$ at which $\special(t) \in [t/n - c, t/n + c]$ is $\Omega(m)$.
\label{lem:specialband} 
\end{lemma}
\begin{proof}
Let $X(t)$ be the event that, at time step $t$, we have both that $\special(t) \not\in [t/n - c, t/n + c]$, and that at least 90\% of bins have loads in the range $[t/n - c, t/n + c]$. Define $\Delta(t) \coloneq |\special(t) - t/n| - |\special(t - 1) - (t - 1)/n|$. 

We will prove that, 
\begin{equation}
\E[\Delta(t) \mid X(t)] \le -0.5/n,
\label{eq:Deltadec1}
\end{equation}
and that
\begin{equation}
\E[\Delta(t) \mid \overline{X}(t)] \le 3/n.
\label{eq:Deltadec2}
\end{equation}

Before we prove \eqref{eq:Deltadec1} and \eqref{eq:Deltadec2}, let us first take them for granted, and see how to prove the lemma at hand. If  \eqref{eq:Deltadec1} and \eqref{eq:Deltadec2} hold, then it follows that
$$\E\left[\sum_t \Delta(t)\right] \le \sum_t \left(\Pr[\overline{X}(t)]\cdot 3/n -  \Pr[X(t)] \cdot 1 / (2n)\right).$$
But, by definition, $\sum_t \Delta(t) = |\special(m) - m/n| \ge 0$. Thus
$$\sum_t \Pr[\overline{X}(t)]\cdot 3/n  \ge \sum_t \Pr[X(t)] \cdot 1/(2n),$$
which implies that
$$\sum_t \Pr[\overline{X}(t)] \ge \Omega(m).$$
By a union bound,
$$\sum_t \Pr[\overline{X}(t)] \le \sum_t \Pr[Y(t)] + \sum_t \Pr[Z(t)],$$
where $Y(t)$ is the event that $\special(t) \in [t/n - c, t/n + c]$ and $Z(t)$ is the event that less than 90\% of bins have loads in $[t/n - c, t/n + c]$ at time $t$. By Lemma \ref{lem:bandhigh}, $\Pr[Z(t)] \le 1/\poly(n)$, so we can conclude that 
$$\sum_t \Pr[Y(t)] \ge \Omega(m),$$
which implies the lemma.

It remains to prove \eqref{eq:Deltadec1} and \eqref{eq:Deltadec2}. To prove \eqref{eq:Deltadec1}, let us condition on event $X(t)$, and consider two cases, depending on whether $\special(t) < t/n - c$ or $\special(t) > t/n + c$. 

Suppose event $X(t)$ occurs and that $\special(t) < t/n - c$. If the $t$-th insertion has one hash equal to the special bin and the other equal to one of the $\ge 90\%$  of bins with loads in $[t/n - 1, t/n + c]$, then $\special(t)$ gets incremented during the insertion. Such an increment occurs with probability at least $2 \cdot 0.9 / n$, and contributes $-1$ to $\Delta(t)$. The insertion also increases the quantity $t/n$ by $1/n$, which adds $1/n$ to $\Delta(t)$. It follows that 
\begin{equation}\E[\Delta(t) \mid X(t) \text{ and }\special(t) < t/n - c] \le -1 \cdot (2 \cdot 0.9 / n) + 1/n \le -1/(2n).
\label{eq:Deltaa}
\end{equation}

Next, suppose that event $X(t)$ occurs and that $\special(t) > t/n + c$. In this case, the only way that $\special(t)$ can increment during the $t$-th insertion is if the insertion has one hash equal to $X(t)$ and the other hash equal to one of the $\le 10\%$ of bins with loads $> tn + c$. Such an increment occurs with probability at most $2 \cdot 0.1 / n$, and contributes $1$ to $\Delta(t)$. The insertion also increases the quantity $t/n$ by $1/n$, which adds $-1/n$ to $\Delta(t)$. It follows that 
\begin{equation}\E[\Delta(t) \mid X(t) \text{ and }\special(t) > t/n + c] \le 0.2/n - 1/n \le -1/(2n).
\label{eq:Deltab}
\end{equation}

Combined, \eqref{eq:Deltaa} and \eqref{eq:Deltab} imply \eqref{eq:Deltadec1}. Finally, to prove \eqref{eq:Deltadec2}, suppose that event $X(t)$ does not occur. In this case, we can still use the following argument. The probability of $\special(t)$ incrementing during insertion $t$ is at most $2/n$, since at least one of the insertion's hashes would need to be the special bin; the expected contribution to $\Delta(t)$ from $\special(t)$ changing is therefore at most $2/n$. The insertion also increases $t/n$ by $1/n$, and this change contributes at most $1/n$ to $\Delta(t)$. It follows that 
$$\E[\Delta(t)] \le 2/n + 1/n = 3/n,$$
as required by \eqref{eq:Deltadec2}.
\end{proof}

We now get to the main argument in the section, where we lower bound the expected number of critical ties that are created over time. To simplify the argument, we begin by describing it in terms of the number of critical ties that occur during a random subwindow: 
\begin{lemma}
Suppose $m \le n^{2 - \Omega(1)} \cap \omega(n)$. Let $c$ be a sufficiently large positive constant, an let $c'$ be a sufficiently large constant as a function of $c$. Consider a time window of the form $[t, t + cn]$, where $t$ is selected uniformly at random from $t \in [3cn, m - cn]$.  The expected number of critical ties that are created during that window is $\Omega(n)$. 
\end{lemma}
\begin{proof}
It suffices to show that, with probability $\Omega(1)$, there are $\Omega(n)$ bins whose heights at time $t$ are less than $\special(t) - 1$, but whose heights at time $t + cn$ are at least $\special(t + cn)$. 

By Lemma \ref{lem:specialband} (and the fact that $t$ is uniformly random across $m - o(m)$ values), we have with probability $\Omega(1)$ that $\special(t) \in [t/n - c, t/n + c]$. Conditioned on this, during the next $c'n$ insertions there is probability $(1 - 2/n)^{c'n} = \Omega(1)$ that no inserted balls hash to the special bin, and therefore that $\special(t + c'n) = \special(t) \le t/n +c$. Putting together the events, we have with probability $\Omega(1)$ that
\begin{equation}\special(t) \ge t/n -c \text{ and }\special(t + c'n) \le t/n + c.
\label{eq:twoevents1}
\end{equation}

By Lemma \ref{lem:belowband}, there exists some $\gamma = \Omega(1)$ such that, with high probability in $n$ that at least $\gamma \cdot n$ bins have loads below $t/n - c$ at time $t$. Supposing $c'$ is sufficiently large as a function of $\gamma$, we have by Lemma \ref{lem:bandhigh} that, with high probability in $n$, at least a $(1 - \gamma/2)n$ bins have loads above $t/n + c$ at time $t + c'n$. It follows that, with high probability in $n$, there exist at least $\gamma n /2$ bins $j$ with 
$$\load{t}{j} \le t/n - c \text{ and } \load{t + c'n}{j} \ge t/n + c.$$

Combining this with \eqref{eq:twoevents1}, we can conclude that with probability $\Omega(1)$, there are at least $\gamma n/2 = \Omega(n)$ bins whose loads at time $t$ are smaller than $\special(t)$ but whose loads at time $t + c'n$ are greater than $\special(t + c'n)$. Every such bin must have experienced a critical tie at some point in the time interval $[t, t + c'n]$. This completes the proof. 
\end{proof}

As a corollary, we can lower bound the expected overall number of critical ties that occur.
\begin{corollary}
Suppose $m \le n^{2 - \Omega(1)} \cap \omega(n)$. The expected number of critical ties created across all insertions is $\Omega(m)$.
\label{cor:numties}
\end{corollary}
\begin{proof}
Let $c$ and $c'$ be as in the previous lemma. We know that for $t$ selected uniformly at random from $[3cn, m - c'n]$, the expected number of critical ties created in the window $[t, t + c'n]$ is $\Omega(n)$. This sampling process captures each critical tie with probability at most $(c'n) / (m - 3cn - c'n) = O(n/m)$. The total expected number of critical ties that get created over all insertions must therefore be at least $\Omega(n \mu) = \Omega(m)$. 
\end{proof}

Finally, we can complete the proof of Proposition \ref{prop:greedylower}.
\begin{proof}[Proof of Proposition \ref{prop:greedylower}]
For $m = O(n)$, the proposition is trivial. For $m = \omega(n)$, the proposition follows from Lemma \ref{lem:tiestorecourse} and Corollary \ref{cor:numties}.
\end{proof}

\section{McDiarmid's Inequality}\label{app:mcdiarmid}

In several sections of this paper, it is helpful to make use of the following inequality due to McDiarmid \cite{mcdiarmid1989method}:

\begin{thm}[McDiarmid's Inequality \cite{mcdiarmid1989method}]
Say that a function $F(X_1, \ldots, X_n)$ is $L$-Lipschitz if changing the value of any single argument $X_i$ never changes $F$ by more than $L$. Let $X_1, \ldots, X_n$ be independent random variables, and let $F$ be an $L$-Lipschitz function. Then, for all $k \ge 1$,
$$\Pr[F(X_1, \ldots, X_n) - \E[F(X_1, \ldots, X_n)] \ge k \sqrt{n}] \le e^{-\Omega(k^2)}.$$
\label{thm:mcdiarmid}
\end{thm}
\end{document}